\documentclass[sigconf, nonacm]{acmart}
\newcommand\vldbdoi{XX.XX/XXX.XX}
\newcommand\vldbpages{XXX-XXX}
\newcommand\vldbvolume{XX}
\newcommand\vldbissue{X}
\newcommand\vldbyear{XXXX}
\newcommand\vldbauthors{\authors}
\newcommand\vldbtitle{\shorttitle}
\newcommand\vldbavailabilityurl{https://github.com/wizicer/SettleFL}
\newcommand\vldbpagestyle{plain}

\usepackage{listings}

\usepackage{amsthm}
\newtheorem{definition}{Definition}
\newtheorem{theorem}{Theorem}

\usepackage{array}
\usepackage{pifont}
\usepackage{xspace}
\usepackage{xparse}

\usepackage[ruled,vlined,linesnumbered]{algorithm2e}
\usepackage{multirow}
\usepackage{subcaption}

\usepackage{bm}

\makeatletter
\renewcommand\paragraph{\def\@toclevel{4}
    \@startsection{paragraph}{4}{0pt}
    {-.5\baselineskip \@plus -2\p@ \@minus -.2\p@}
    {-3.5\p@}
{\ACM@NRadjust{\@parfont\bfseries}}}
\makeatother

\newcommand{\sys}{\textsf{\kern0.5pt SettleFL}\xspace}

\definecolor{lightgray}{gray}{0.80}

\NewDocumentCommand{\bcfl}{}{\ensuremath{\texttt{BCFL}}\xspace}
\NewDocumentCommand{\piccfl}{}{\ensuremath{\textsf{ccFL}}\xspace}
\NewDocumentCommand{\ccfl}{}{\ensuremath{\textsf{SettleFL-CC-Sol}}\xspace}
\NewDocumentCommand{\zkccfl}{}{\ensuremath{\textsf{SettleFL-CC}}\xspace}
\NewDocumentCommand{\zkrfl}{}{\ensuremath{\textsf{SettleFL-CP}}\xspace}

\NewDocumentCommand{\agg}{}{\ensuremath{\mathcal{A}}\xspace}
\NewDocumentCommand{\pars}{}{\ensuremath{\mathcal{P}}\xspace}

\NewDocumentCommand{\opjobcreation}{}{\ensuremath{\operatorname{create}}\xspace}
\NewDocumentCommand{\opcommit}{}{\ensuremath{\operatorname{commit}}\xspace}
\NewDocumentCommand{\optransition}{}{\ensuremath{\operatorname{transition}}\xspace}
\NewDocumentCommand{\opfinalize}{}{\ensuremath{\operatorname{finalize}}\xspace}
\NewDocumentCommand{\opdistribute}{}{\ensuremath{\operatorname{distribute}}\xspace}
\NewDocumentCommand{\opdistributeo}{}{\ensuremath{\operatorname{distribute-1s}}\xspace}
\NewDocumentCommand{\opdistributem}{}{\ensuremath{\operatorname{distribute-ms}}\xspace}

\NewDocumentCommand{\opchallenge}{}{\ensuremath{\operatorname{challenge}}\xspace}
\NewDocumentCommand{\opcounter}{}{\ensuremath{\operatorname{counter}}\xspace}

\newcommand{\stepscale}{1.35}
\newcommand{\stepsize}{\dimexpr \stepscale\ht\strutbox\relax}

\newcommand{\step}[2][1]{
    \raisebox{-0.8ex}{
        \includegraphics[height=\dimexpr #1\stepsize\relax]{step-#2.pdf}
    }
}

\SetCommentSty{MyCmtStyle}
\SetKwComment{Comment}{$\textcolor{gray}{\triangleright}$ }{}
\SetKwComment{CommentDown}{$\textcolor{gray}{\triangledown}$ }{}
\DeclareRobustCommand{\CommentFunction}[1]{
    \unskip\hfill{\textcolor{gray}{\itshape$\triangleright$~#1}}
}

\newcommand{\hagg}{\mathsf{h}_{\agg}}
\newcommand{\salt}{\rho}

\newcommand{\lastround}{r_{\text{last}}}
\newcommand{\onlyOwner}{\textit{onlyOwner}}
\newcommand{\unlockTime}{t_{\mathsf{unlock}}}
\newcommand{\status}{\tau}
\newcommand{\statusCommitted}{\tau_\textsf{Committed}}
\newcommand{\statusRewardInit}{\tau_\textsf{RewardInit}}
\newcommand{\statusChallenged}{\tau_\textsf{Challenged}}
\newcommand{\statusPaid}{\tau_\textsf{Distributed}}
\newcommand{\statusPaying}{\tau_\textsf{Distributing}}
\newcommand{\verify}[1]{\textit{verify}_{\mathsf{#1}}}
\newcommand{\proofpi}[1]{\pi_{\mathsf{#1}}}
\newcommand{\mode}{\gamma}
\newcommand{\modeProof}{\gamma_\textsf{Proof}}
\newcommand{\modeChallenge}{\gamma_\textsf{Challenge}}

\newcommand{\pubpar}[1]{\underline{#1}}

\newcommand{\pubkeyA}{\mathbf{A}}
\newcommand{\gadgetname}[1]{\textsc{#1}}

\newcommand{\constrain}[1]{\ensuremath{\langle #1 \rangle}}

\begin{document}
\title{SettleFL: Trustless and Scalable Reward Settlement Protocol for Federated Learning on Permissionless Blockchains}

\author{Shuang Liang}
\affiliation{
    \institution{Shanghai Jiao Tong University}
}
\email{liangshuangde@sjtu.edu.cn}

\author{Yang Hua}
\affiliation{
    \institution{Queen's University Belfast}
}
\email{Y.Hua@qub.ac.uk}

\author{Linshan Jiang}
\affiliation{
    \institution{National University of Singapore}
}
\email{linshan@nus.edu.sg}

\author{Peishen Yan}
\affiliation{
    \institution{Shanghai Jiao Tong University}
}
\email{peishenyan@sjtu.edu.cn}

\author{Tao Song}
\authornote{Corresponding author.}
\affiliation{
    \institution{Shanghai Jiao Tong University}
}
\email{songt333@sjtu.edu.cn}

\author{Bin Yao}
\authornotemark[1]
\affiliation{
    \institution{Shanghai Jiao Tong University}
}
\email{yaobin@cs.sjtu.edu.cn}

\author{Haibing Guan}
\affiliation{
    \institution{Shanghai Jiao Tong University}
}
\email{hbguan@sjtu.edu.cn}

\renewcommand{\shortauthors}{Liang \etal}

\begin{abstract}
    In open Federated Learning (FL) environments where no central authority exists, ensuring collaboration fairness relies on decentralized reward settlement, yet the prohibitive cost of permissionless blockchains directly clashes with the high-frequency, iterative nature of model training.
    Existing solutions either compromise decentralization or suffer from scalability bottlenecks due to linear on-chain costs.
    To address this, we present \sys{}, a trustless and scalable reward settlement protocol designed to minimize total economic friction by offering a family of two interoperable protocols.
    Leveraging a shared domain-specific circuit architecture, \sys{} offers two interoperable strategies:
    (1) a \emph{Commit-and-Challenge} variant that minimizes on-chain costs via optimistic execution and dispute-driven arbitration, and
    (2) a \emph{Commit-with-Proof} variant that guarantees instant finality through per-round validity proofs.
    This design allows the protocol to flexibly adapt to varying latency and cost constraints while enforcing rational robustness without trusted coordination.
    We conduct extensive experiments combining real FL workloads and controlled simulations.
    Results show that \sys{} remains practical when scaling to 800 participants, achieving substantially lower gas cost.
\end{abstract}

\maketitle
\pagestyle{\vldbpagestyle}
\begingroup\small\noindent\raggedright\textbf{PVLDB Reference Format:}\\
\vldbauthors. \vldbtitle. PVLDB, \vldbvolume(\vldbissue): \vldbpages, \vldbyear.\\
\href{https://doi.org/\vldbdoi}{doi:\vldbdoi}
\endgroup
\begingroup
\renewcommand\thefootnote{}\footnote{\noindent
    This work is licensed under the Creative Commons BY-NC-ND 4.0 International License. Visit \url{https://creativecommons.org/licenses/by-nc-nd/4.0/} to view a copy of this license. For any use beyond those covered by this license, obtain permission by emailing \href{mailto:info@vldb.org}{info@vldb.org}. Copyright is held by the owner/author(s). Publication rights licensed to the VLDB Endowment. \\
    \raggedright Proceedings of the VLDB Endowment, Vol. \vldbvolume, No. \vldbissue\
    ISSN 2150-8097. \\
    \href{https://doi.org/\vldbdoi}{doi:\vldbdoi} \\
}\addtocounter{footnote}{-1}\endgroup

\ifdefempty{\vldbavailabilityurl}{}{
    \vspace{.3cm}
    \begingroup\small\noindent\raggedright\textbf{PVLDB Artifact Availability:}\\
    The source code, data, and/or other artifacts have been made available at \url{\vldbavailabilityurl}.
    \endgroup
}

\section{Introduction}
\label{sec:introduction}
\begin{figure*}[t]
    \centering
    \includegraphics[width=\textwidth]{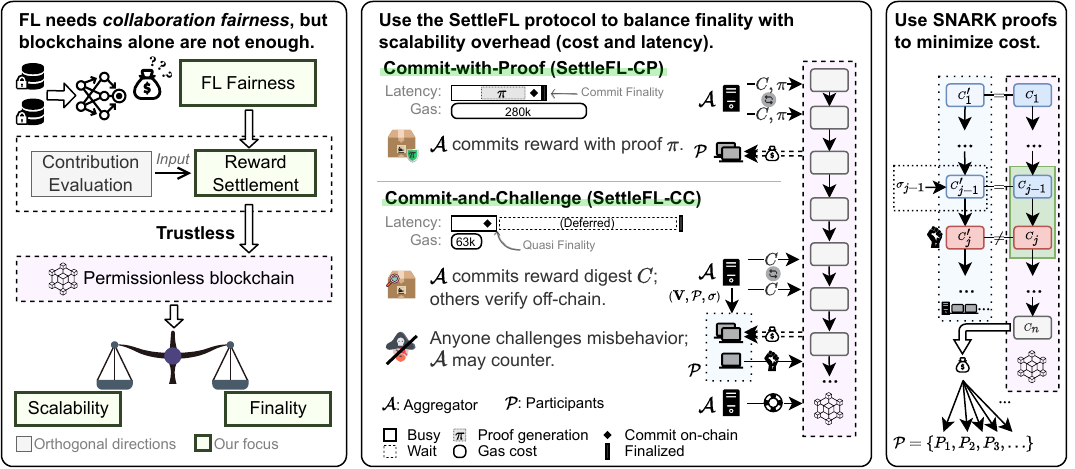}
    \caption{Overview of \sys{} protocol family.
        \textnormal{Our protocol resolves the conflict between costly on-chain execution and high-frequency FL updates.
            It navigates the scalability-finality trade-off through two variants: \emph{Commit-and-Challenge} for minimal cost, and \emph{Commit-with-Proof} for instant finality.
            A shared SNARK architecture flexibly enforces correctness by either guaranteeing validity upfront or resolving disputes with minimal on-chain overhead.
        }
    }
    \Description{
        The figure illustrates the design and components of the SettleFL protocol across three panels:

        **Left Panel: Motivation & Context**
        * **Problem Domain:** Focuses on "FL Fairness", explicitly treating "Contribution Evaluation" as an "Input" to "Reward Settlement" (the core focus).
        * **Key Insight:** States that "FL needs *collaboration fairness*, but blockchains alone are not enough."
        * **Metrics:** Highlights the trade-off between "Scalability" and "Finality" in a Trustless environment.
        * **Legend:** Distinguishes between "Orthogonal directions" (grey) and "Our focus" (green/highlighted).

        **Middle Panel: Protocol Strategies (SettleFL-CP vs SettleFL-CC)**
        * **Commit-with-Proof (SettleFL-CP):**
        * **Mechanism:** Aggregator commits reward with proof \pi.
        * **Characteristics:** Commit Finality, but higher cost and latency.
        * **Commit-and-Challenge (SettleFL-CC):**
        * **Mechanism:** Aggregator commits reward digest C; others verify off-chain.
        * **Interactive Process:** Anyone challenges misbehavior; Aggregator may counter.
        * **Characteristics:** Quasi Finality with low cost/latency, but depends on an interactive dispute resolution period.
        * **SettleFL Goal:** Use the SettleFL protocol to balance finality with scalability overhead (cost and latency).

        **Right Panel: Challenge Mechanism Detail**
        * **Objective:** "Use SNARK proofs to minimize cost" during disputes.
        * **Process:** Depicts a step-by-step state comparison between the Prover (Aggregator) and a Challenger.
        * **Divergence Check:**
        * Shows a sequence of states $C$ and $C'$.
        * Identifies the specific step $j$ where the states diverge: matching at $j-1$ ($C'_{j-1} = C_{j-1}$) but differing at $j$ ($C'_j \neq C_j$).
        * **Outcome:** This mechanism allows the protocol to enforce correctness with minimal on-chain computation.
    }

    \label{fig:overview}
    \vspace{5pt}
\end{figure*}

Federated Learning (FL) enables collaborative model training across distributed participants, yet maintaining long-term participation requires ensuring \emph{collaboration fairness}~\cite{huangFederated2024}.
To achieve this, the system is expected to address two distinct challenges: \emph{contribution evaluation} and \emph{reward settlement}.
Contribution evaluation refers to the algorithmic process of quantifying the value of participant updates, often utilizing methods such as Shapley value estimation~\cite{wangPrincipled2020,chenContributions2024}.
Advanced approaches in this domain may further incorporate contribution verifiability mechanisms to cryptographically prove the correctness of these evaluation results~\cite{xingZeroKnowledge2026}.
However, existing works either treat the execution of rewards as a trivial external assumption~\cite{wangPrincipled2020} or restrict rewards to coarse-grained model utility~\cite{lyuCollaborative2020}.
In open FL environments where no central authority exists, this assumption collapses.
Without a credible enforcement mechanism, even the most accurate and verifiable evaluation algorithms become meaningless promises.
Consequently, this work explicitly excludes contribution evaluation from its scope, treating it as an upstream input, and focuses exclusively on the unsolved challenge of \emph{reward settlement}, aiming to ensure that decided rewards are settled strictly and immutably without a trusted coordinator.

Enforcing such settlement without trusted intermediaries necessitates blockchain technology~\cite{nakamotoBitcoin2008}.
However, the \emph{blockchain trilemma}~\cite{renoNavigating2025} dictates that maximizing decentralization and security inherently compromises scalability.
Since open FL strictly requires a trustless environment, permissionless blockchains emerge as the theoretically optimal infrastructure to maximize censorship resistance and universal verifiability.
However, this architectural choice creates a fundamental conflict: the prohibitive cost of on-chain execution directly clashes with the high-frequency, iterative nature of FL~\cite{liuEnhancing2024}.

Existing solutions fail to resolve this tension.
Approaches utilizing permissioned or hybrid architectures~\cite{desaiBlockFLA2021,fanHybrid2021} reduce costs but reintroduce trust assumptions via closed consortia or cross-chain bridges, undermining the censorship resistance required for open collaboration.
Conversely, state channels~\cite{iacobanStateFL2024} suffer from fragmented dispute resolution: their isolated 1-to-1 topology prevents collective security, forcing every participant to individually initiate challenges to secure their funds.
Ultimately, these methods compel participants to \emph{actively} submit transactions for settlement, causing gas costs to scale linearly with the user base.
The core research problem is therefore: \emph{How to design a protocol that enforces reward settlement on permissionless blockchains while minimizing on-chain costs at scale?}

To address this, we present \sys{}, a trustless and scalable reward settlement protocol as illustrated in Figure~\ref{fig:overview}.
Designed to bridge the gap between algorithmic fairness and system practicality, \sys{} creates a flexible design space by introducing two interoperable variants within a single protocol family.
The \emph{Commit-and-Challenge} variant minimizes costs by adopting an optimistic execution model where the aggregator commits only a cryptographic digest of the reward allocation on-chain.
This design allows honest participants to remain passive; they verify data off-chain and are never required to submit transactions unless they detect malpractice.
Conversely, for environments requiring immediate guarantees, the \emph{Commit-with-Proof} variant attaches a validity proof, specifically a Succinct Non-interactive ARgument of Knowledge (SNARK) proof, to every commitment, offering instant finality.

While drawing architectural inspiration from blockchain scaling primitives~\cite{sguanciLayer2021}, \sys{} fundamentally differs by trading generality for domain-specific efficiency.
Unlike general-purpose optimistic rollups that often rely on complex, multi-round interactive dispute games, \sys{} leverages the structured nature of FL rewards to enable a single-shot non-interactive challenge, drastically reducing dispute complexity.
Similarly, unlike standard validity rollups that batch heterogeneous transactions arbitrarily, \sys{} aligns verification strictly with FL training rounds.
This design enables fine-grained per-round settlement without waiting for external block-level aggregation, ensuring that the incentive mechanism matches the iterative rhythm of FL.

Our main contributions are as follows:

\begin{itemize}
    \item
        We propose a protocol family, \sys{}, that addresses the scalability-finality trade-off in blockchain-based FL.
        It provides two interoperable variants, namely an optimistic Commit-and-Challenge variant for minimal cost and a Commit-with-Proof variant for instant finality, allowing applications to select the appropriate security model.

    \item
        We design a highly optimized SNARK circuit architecture that serves as the foundation for our modular verification mechanism.
        By employing shared cryptographic gadgets and a specifically designed commitment scheme, this infrastructure efficiently supports both immediate validity proofs and optimistic challenges, enabling the protocol to scale to large participant sets with constant on-chain complexity.

    \item
        We demonstrate the system's practicality through extensive experiments and a live deployment on the Ethereum Sepolia testnet.
        Results show that \sys{} successfully scales to 800 participants, a scale previously considered impractical for on-chain execution.
        Notably, it settles rewards for 50 training rounds at a total cost of approximately \$4 (USD), effectively removing the economic barrier for decentralized FL incentives.
\end{itemize}

\section{Preliminary}
\label{sec:preliminary}
\begin{figure}[t]
    \centering
    \includegraphics[width=0.4\textwidth]{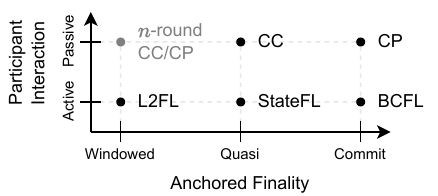}
    \caption{
        Landscape of Settlement Strategies.
        \textnormal{Unlike the active baselines L2FL~\cite{aliefFLB22023}, StateFL~\cite{iacobanStateFL2024}, and BCFL~\cite{xuBAFL2021} which require user participation, \sys{} enables passive interaction via its two variants, \zkccfl{} (CC) and \zkrfl{} (CP).
        Notably, if anchored finality is compromised to the windowed stage, the protocol essentially degrades to an $n$-round batched CC or CP.}
    }

    \Description{
        +-------------+------------------+------------------+------------------+
        | Interaction |     Windowed     |      Quasi       |      Commit      |
        +-------------+------------------+------------------+------------------+
        |   Active    |       L2FL       |      StateFL     |       BCFL       |
        |      |      |        |         |         |        |         |        |
        |      v      |        v         |         v        |         v        |
        |   Passive   |  n-round CC/CP   |        CC        |        CP        |
        +-------------+------------------+------------------+------------------+
        (Anchored Finality Axis -->)
    }
    \label{fig:plot_finality_interaction}
\end{figure}

\subsection{Blockchain-based Federated Learning}
\label{subsec:bcfl}

Blockchain-based Federated Learning (BCFL) integrates blockchain infrastructure into FL workflows to ensure transparency, fairness, and traceability in decentralized training environments.
While BCFL encompasses various mechanisms to enhance collaboration fairness, such as contribution evaluation or verifiable computation, our work specifically focuses on \emph{reward settlement}, utilizing the blockchain as an immutable \emph{incentive layer}.
In this context, smart contracts \emph{enforce} and \emph{automate} the allocation process, ensuring that rewards are distributed transparently according to predefined rules.
However, compared to centralized approaches, a secure permissionless blockchain implementation, referred to as \emph{Layer-1 (L1)}, inevitably introduces significant performance degradation due to the consensus overhead~\cite{hanHow2022}.

\paragraph{Anchored Finality.}
To mitigate L1 performance bottlenecks, existing BCFL frameworks attempt to bypass these limitations through two main approaches.
The first relies on \emph{alternative chains} with reduced overhead, including permissioned blockchains~\cite{xuBAFL2021, wengDeepChain2021}, sidechains~\cite{iacobanStateFL2024, chenFLock2025}, and lightweight independent chains~\cite{aliefFLB22023, difAutoDFL2025}.
While efficient, these architectures typically compromise security by depending on weaker consensus models.
The second approach utilizes \emph{Layer-2 (L2)}~\cite{gangwalSurvey2023} solutions, which we denote as L2FL~\cite{aliefFLB22023}.
L2FL maintains L1 security by enforcing \emph{anchored finality}, where a reward settlement is considered final only after its corresponding state commitment is successfully recorded on the main L1 ledger.
Critically, waiting for this cross-layer synchronization to complete imposes an unavoidable time delay, preventing instant reward confirmation.

\paragraph{Interaction.}
Regardless of the chosen architecture, the efficiency of the system is constrained by the frequency of on-chain operations, defined here as \emph{interaction}.
Distinct from standard communication protocols, every on-chain interaction inherently incurs a direct financial cost (i.e., gas) and requires waiting for settlement latency.
This overhead scales poorly in FL settings, which typically involve a single aggregator but a vast number of participants.
Therefore, as positioned in Figure~\ref{fig:plot_finality_interaction}, a critical objective in BCFL design is to minimize \emph{participant interaction}~\cite{iacobanStateFL2024, wahrstatterOpenFL2024}.
By enforcing a passive model, \sys{} eliminates user-side gas costs and maximizes scalability, accepting a bounded risk of at most one lost reward.

\paragraph{Rollup Mechanisms.}
Among L2 architectures, \emph{rollups} represent a dominant subclass distinguished by their requirement to publish transaction data or state commitments directly to the L1, ensuring data availability.
Mainstream rollups employ two distinct verification paradigms.
Optimistic rollups operate on a presumption of validity, posting state commitments to L1 subject to a fraud-proof challenge window.
This design minimizes immediate computation but necessitates a prolonged dispute period, often lasting days, to ensure security.
Conversely, validity rollups enforce correctness via succinct cryptographic proofs (e.g., SNARKs) that are verified by an L1 smart contract.
While this approach enables faster settlement compared to optimistic designs, it introduces substantial computational overhead for proof generation.

\subsection{SNARKs}

SNARKs are cryptographic proof systems that allow a prover to convince a verifier of the validity of a statement in a non-interactive and succinct way.
It is defined by the tuple $\textsf{SNARK} = (\mathsf{Setup}, \mathsf{Prove},\break \mathsf{Verify})$.
The $\mathsf{Setup}(1^\lambda, \mathcal{C})$ algorithm outputs public parameters $\mathsf{pp}$ for circuit $\mathcal{C}$.
The $\mathsf{Prove}(\mathsf{pp},x,w)\to \pi$ algorithm enables the prover to generate a succinct proof $\pi$ that $\mathcal{C}(x,w)=1$.
The $\mathsf{Verify}(\mathsf{pp},x,\pi)\break\to\{\mathsf{accept},\mathsf{reject}\}$ algorithm allows the verifier to check the validity of $\pi$ on input $x$.
A SNARK satisfies \emph{completeness} (valid proofs are always accepted) and \emph{soundness} (false statements cannot be proven).
In what follows, we instantiate our proofs using Groth16~\cite{grothSize2016}, a SNARK that is well supported on blockchains such as Ethereum.

\medskip
\noindent\textbf{Arithmetic circuits.}
In SNARKs, a statement is represented as an arithmetic circuit expressed by algebraic constraints, and the prover must show that the witness satisfies all constraints with respect to the public input.
The efficiency of proof generation depends on the number of constraints, which can be greatly reduced when the underlying operations are SNARK-friendly, i.e., encodable with few algebraic constraints.

\section{System Overview}
\label{sec:system}

\subsection{System Model}
\label{sec:system_model}

We consider an FL setting involving an aggregator \agg and a set of participants $\mathcal{P} = \{P_1, \ldots, P_N\}$, who collaboratively train a global machine learning model.
In each round $r$, the aggregator distributes the current global model to all participants, each of whom trains it locally on private data and returns an update to \agg.
The aggregator then aggregates these updates to produce a new global model and evaluates the contribution of each participant.
Based on this evaluation, a reward allocation vector $\mathbf{v}^r = (v_1^r, \ldots, v_N^r)$ is derived, where $v_i^r$ denotes the reward associated with participant $P_i$ in round $r$.

To ensure accountability and enable verifiable reward settlement, the system incorporates a permissionless blockchain $\mathbb{B}$ that serves as a public ledger for recording reward commitments and related state transitions.
The aggregator interacts with $\mathbb{B}$ to publish these commitments, while participants can reference on-chain records to verify or challenge the reported results.

\subsection{Threat Model}
\label{sec:threat}

We consider an adversary $\agg^\ast$ who controls the aggregator and may arbitrarily deviate from the protocol to manipulate reward settlement.
$\agg^\ast$ can observe all messages and submit inconsistent or fraudulent commitments, but is assumed unable to break cryptographic primitives or compromise the integrity of the blockchain.
Participants may also behave maliciously, for example, by submitting incorrect updates or issuing dishonest challenges.
The blockchain $\mathbb{B}$ is assumed to provide consensus security and eventual consistency, ensuring that once a commitment is on-chain, it cannot be reverted or altered.
Participants are assumed rational and choose actions that maximize their own reward given the protocol's payoffs; they tolerate at most one missing reward and, upon detecting aggregator misbehavior, exit the training, securing rewards already recorded on chain and avoiding further loss.

We focus on four representative threats.
(1) \textbf{Commitment reversal}: $\agg^\ast$ may locally announce a correct commitment and then publish a fraudulent one excluding certain contributions.
(2) \textbf{Refusal to commit}: $\agg^\ast$ withholds the final commitment on-chain, blocking participants from receiving rewards.
(3) \textbf{Reward withholding}: $\agg^\ast$ commits the reward but deliberately withholds the final transaction that distributes funds.
(4) \textbf{Stale challenges}: malicious $\pars^\ast$ may challenge a stale commitment in order to disrupt the system or extract undeserved rewards.\footnote{
    Our model focuses on reward enforceability in blockchain-based FL.
    We do not consider orthogonal concerns such as participant-level privacy leakage, model misuse or theft, delegated model training with misreported contributions~\cite{hallajiDecentralized2024}, or aggregator-participant collusion~\cite{hossainAdRoFL2025}.
}

\newcommand{\wgas}{w_g}
\newcommand{\wproto}{w_p}
\newcommand{\wcomp}{w_c}

\newcommand{\Cgas}{C_g}
\newcommand{\Tproto}{T_p}
\newcommand{\Tcomp}{T_c}

\newcommand{\Cbase}{C_b}
\newcommand{\Cverif}{C_v}
\newcommand{\Twin}{\Delta_w}
\newcommand{\Tproof}{\Delta_p}

\newcommand{\Pival}{\Pi_{\text{val}}}
\newcommand{\Piopt}{\Pi_{\text{opt}}}

\subsection{Problem Formulation}
\label{sec:problem_formulation}

We aim to design a protocol $\Pi$ that ensures trustless and scalable reward settlement under the threat model described in Section~\ref{sec:threat}.
We formulate this as a stronger setting of \emph{Practical Scalable BCFL}, requiring the protocol to support $N \ge 200$ participants, operate within realistic gas limits, ensure robustness without trusted coordination, and enable participants to passively receive rewards.

To rigorously define the goal, we first establish the definition of BCFL reward enforceability.

\begin{definition}[BCFL]
    \label{def:bcfl}
    Consider an aggregator \agg, participants $\mathcal{P}$, and a blockchain $\mathbb{B}$.
    A protocol $\Pi$ is a BCFL protocol if, for every round $r$, it specifies a reward vector $\mathbf{v}^r$ and enforces that the payouts are deterministically settled on $\mathbb{B}$.
    We define such settlement as having \emph{anchored finality} when the allocations implied by $\mathbf{v}^r$ become uniquely determined and are guaranteed to be distributed on $\mathbb{B}$ for all eligible participants.
\end{definition}

To maximize system availability and alleviate the operational burden on participants, we enforce a strict constraint of \emph{passive interaction}, defined as allowing participants to receive rewards without initiating direct interactions with the blockchain.
Under this constraint, we formalize the efficiency objective as a minimization of the \emph{Total Economic Friction}.

\begin{definition}[Total Economic Friction]
    \label{def:friction}
    Let $\Pi$ denote a BCFL protocol that supports passive interaction.
    The optimization objective is to minimize the weighted sum of explicit financial costs, protocol-enforced delays, and computational latencies:
    \begin{equation}
        \label{eq:friction_function}
        \min_{\Pi \in \text{BCFL}} \mathcal{F}(\Pi) = \wgas \cdot \Cgas(\Pi) + \wproto \cdot \Tproto(\Pi) + \wcomp \cdot \Tcomp(\Pi),
    \end{equation}
    where:
    \begin{itemize}
        \item $\Cgas(\Pi)$ represents the on-chain execution cost.
        \item $\Tproto(\Pi)$ denotes the mandatory waiting period enforced by the protocol, such as a dispute window.
        \item $\Tcomp(\Pi)$ denotes the physical runtime required to generate validity proofs.
        \item $\wgas, \wproto, \wcomp$ are task-specific sensitivity weights.
    \end{itemize}
\end{definition}

Note that throughout this paper, \emph{settlement} refers to the cryptographic finalization of state $\mathbf{v}^r$, while \emph{distribution} denotes the subsequent execution of asset transfers.
\sys{} focuses on optimizing settlement as the essential prerequisite for distribution.

\subsection{Solution Overview}
\label{sec:solution_overview}

To minimize $\mathcal{F}$ defined in Equation~\ref{eq:friction_function}, we analyze the behavior of the cost function at the boundaries of the design space.
Since hybrid approaches incur verification costs without removing the need for a dispute window, they fail to minimize either the financial or temporal friction.
Therefore, the design space collapses into two distinct variants, each zeroing out a specific latency component.

\paragraph{Boundary Strategies.}
We define the two boundary strategies, $\Piopt$ and $\Pival$, based on which latency term they eliminate:

\begin{itemize}
    \item \textbf{Optimistic Strategy ($\Piopt$):}
        This strategy prioritizes execution efficiency by eliminating computational latency ($\Tcomp \to 0$).
        However, to maintain security without verification, it must enforce a substantial protocol delay, denoted as $\Twin$.
        With small overhead gas denoted as $\Cbase$, the total friction is thus dominated by the waiting cost:
        \[
            \mathcal{F}(\Piopt) \approx \wgas \cdot \Cbase + \wproto \cdot \Twin.
        \]

    \item \textbf{Validity Strategy ($\Pival$):}
        This strategy prioritizes immediate finality by eliminating protocol delay ($\Tproto \to 0$).
        To achieve this, it incurs the physical time for proof generation, denoted as $\Tproof$, and the higher gas cost $\Cverif$ for verification.
        The total friction is thus dominated by the verification overhead:
        \[
            \mathcal{F}(\Pival) \approx \wgas \cdot \Cverif + \wcomp \cdot \Tproof.
        \]
\end{itemize}

A rational system design should adopt the Validity Strategy if and only if it yields a lower total friction than the Optimistic Strategy, i.e., $\mathcal{F}(\Pival) < \mathcal{F}(\Piopt)$.
Substituting the expressions above:
\[
    \wgas \cdot \Cverif + \wcomp \cdot \Tproof < \wgas \cdot \Cbase + \wproto \cdot \Twin.
\]
Rearranging the terms to group the cost components yields the \emph{Fundamental \sys{} Inequality}:

\begin{equation}
    \label{eq:decision_boundary}
    \underbrace{\wproto \cdot \Twin}_{\text{Opportunity Cost}} > \underbrace{\wgas \cdot (\Cverif - \Cbase) + \wcomp \cdot \Tproof}_{\text{Verification Overhead}}.
\end{equation}

This inequality provides the rigorous logic for variant selection: the system should switch to validity strategy only when the implicit \emph{opportunity cost} of waiting exceeds the explicit \emph{verification overhead}.

\paragraph{The Proposed Solutions.}
Guided by this derivation, we propose \sys, which comprises two interoperable protocols and explicitly implements these two boundary strategies to cover the full range of FL scenarios:
(1)~\zkccfl, which implements $\Piopt$:
Designed for the regime where the inequality fails (e.g., cost-sensitive tasks).
This protocol utilizes a \emph{Commit-and-Challenge} mechanism to minimize on-chain costs, accepting a challenge window to secure the system.
(2)~\zkrfl, which implements $\Pival$:
Designed for the regime where the inequality holds (e.g., time-sensitive tasks).
This protocol utilizes a \emph{Commit-with-Proof} mechanism to provide immediate settlement, justifying the computational and gas premium for instant liquidity.

Both protocols are built upon a shared set of off-chain circuit gadgets and on-chain functions, ensuring that \sys remains modular.
\begin{figure}[t]
    \centering
    \includegraphics[]{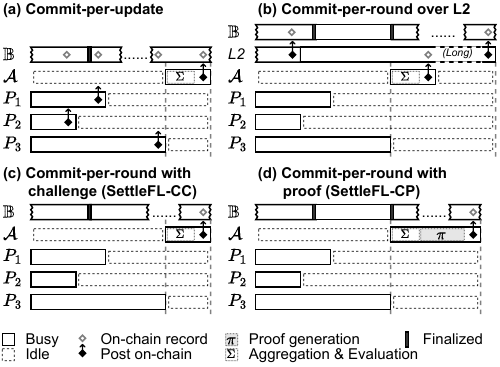}
    \caption{
        Timeline Analysis of Latency.
        \textnormal{By synchronizing settlement with the aggregation cycle, \sys{} eliminates extrinsic delays, bounding latency to the theoretical minimum.}
    }

    \Description{
        Timeline diagrams comparing four commitment strategies.

        (a) Commit-per-update:
        - Visual: Frequent interaction.
        - Flow: Participants compute -> Immediate update to Blockchain.
        - Blockchain Timeline: Dotted with many black diamonds (records).

        (b) Commit-per-round over L2:
        - Visual: Two-layer interaction.
        - Flow: Participants -> L2 (Long block time) -> Main Blockchain.
        - Difference: Adds an intermediate L2 timeline above the chain.

        (c) Commit-per-round (Standard):
        - Flow: Participants compute -> Aggregator performs Aggregation (Sigma) -> Single commit to Blockchain.
        - Visual: A distinct aggregation block precedes the on-chain record.

        (d) Commit-per-round with proof (Proposed):
        - Flow: Participants compute -> Aggregator performs Aggregation (Sigma) -> **Proof Generation (Pi)** -> Single commit.
        - Visual: Adds a gray "Proof generation" block after aggregation, delaying the on-chain record compared to (c).
    }

    \label{fig:finality}
\end{figure}

\paragraph{Settlement Timeline.}
To intuitively illustrate how \sys{} minimizes $\mathcal{F}$ compared to prior paradigms, Figure~\ref{fig:finality} presents a timeline analysis of settlement latency.
Traditional approaches typically suffer from either high interaction overhead due to frequent commit-per-update mechanisms (Figure~\ref{fig:finality}a), or extrinsic settlement delays imposed by generic L2 block times (Figure~\ref{fig:finality}b).
In contrast, both variants of \sys{} synchronize settlement strictly with the FL aggregation cycle.
Whether issuing a single per-round commitment instantly (Figure~\ref{fig:finality}c) or waiting solely for proof generation (Figure~\ref{fig:finality}d), the design eliminates all other extrinsic protocol delays, bounding the total latency to the theoretical minimum.

\section{Protocol Design}
\label{sec:protocol}

We propose a protocol family whose key distinction lies in the finality mechanism.
This design constitutes the primary contribution of this paper.

\subsection{Architecture}
\label{sec:architecture}
\begin{figure}[ht]
    \centering
    \includegraphics[width=0.4\textwidth]{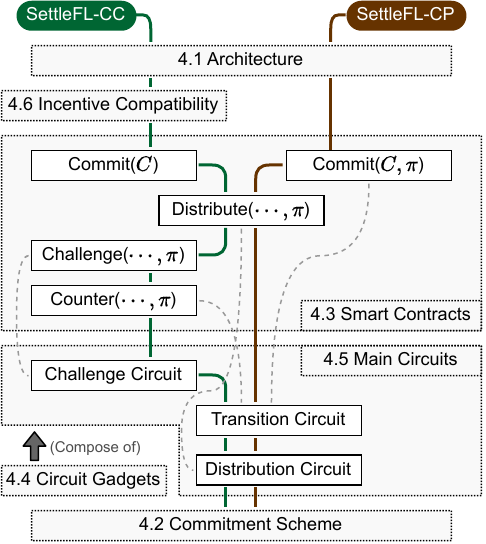}
    \caption{System architecture overview.
        \textnormal{Shared cryptographic primitives underpin the architecture, enabling specialized modules to enforce distinct security assumptions.}
    }

    \Description{
        **System Modules:**
        - 4.1 Architecture
        - 4.2 Commitment Scheme
        - 4.3 Smart Contracts
        - 4.4 Circuit Gadgets
        - 4.5 Main Circuits
        - 4.6 Incentive Compatibility

        **Actors:**
        - SettleFL-CC
        - SettleFL-CP

        **Smart Contract Functions:**
        - Commit($C$)
        - Commit($C, \pi$)
        - Distribute($\cdots, \pi$)
        - Challenge($\cdots, \pi$)
        - Counter($\cdots, \pi$)

        **Circuit Components:**
        - Challenge Circuit
        - Transition Circuit
        - Distribution Circuit
        - (Compose of)

        **Visual Relationships:**
        - The diagram is layered from Actors (top) to Smart Contracts (middle) to Circuits (bottom).
        - CC (Green) connects to Commit($C$) and Incentive Compatibility.
        - CP (Brown) connects to Commit($C, \pi$), Distribute, and the Commitment Scheme.
        - The Transition Circuit serves as a dependency, connecting via arrows to the Challenge Circuit and Distribution Circuit.
    }

    \label{fig:architecture}
\end{figure}

To implement the proposed protocol family, we design a hierarchical architecture, illustrated in Figure~\ref{fig:architecture}.
The architecture is built upon a common base used by both protocols.
At the lowest level, a unified \emph{commitment scheme} (Section~\ref{sec:commitment-scheme}) serves as the foundation for state representation.
Building on this, we define a set of cryptographic \emph{circuit gadgets} (Section~\ref{sec:gadgets}) that act as the construction primitives for all circuits.
These primitives are assembled into the core \emph{transition} and \emph{distribution circuits} (Section~\ref{sec:circuits}), while the smart contract (Section~\ref{sec:contract}) shares the fundamental \emph{distribute} logic.

Leveraging this shared hierarchy, the architecture branches with specialized modules to enforce specific security assumptions.
\zkrfl{} utilizes a verifiable \emph{commit} function that strictly verifies the validity of the reward settlement transition.
In contrast, \zkccfl{} employs an optimistic contract designed for the \emph{Commit-and-Challenge} workflow.
This branch extends the shared circuits with a specialized \emph{challenge circuit}, reuses the transition circuit as a \emph{counter circuit}, and incorporates an additional \emph{incentive compatibility} layer (Section~\ref{sec:incentive}) to secure the dispute window against malicious actors.

\paragraph{Adaptation from Rollup Mechanisms.}
While \sys{} draws architectural inspiration from the generic Optimistic and Validity Rollups introduced in Section~\ref{subsec:bcfl}, a direct application of standard rollup stacks would restrict the system to \emph{windowed finality}.
As positioned in the settlement landscape in Figure~\ref{fig:plot_finality_interaction}, naively increasing the commitment frequency to a per-round basis to achieve \emph{quasi} or \emph{commit} finality is impractical, as standard rollups are optimized for general-purpose workloads, incurring prohibitive gas costs and verification latency at high frequencies~\cite{aliefFLB22023}.

To overcome this, \sys{} trades general-purpose execution for domain-specific efficiency tailored to FL workloads.
For the validity workflow, our specialized commitment scheme significantly accelerates proof generation by streamlining circuit logic.
For the optimistic workflow, it enables dispute resolution within a single interaction round, bypassing the complex multi-round interactive fraud proofs typical in traditional optimistic rollups.
\begin{figure}[t]
    \centering
    \includegraphics[width=\columnwidth]{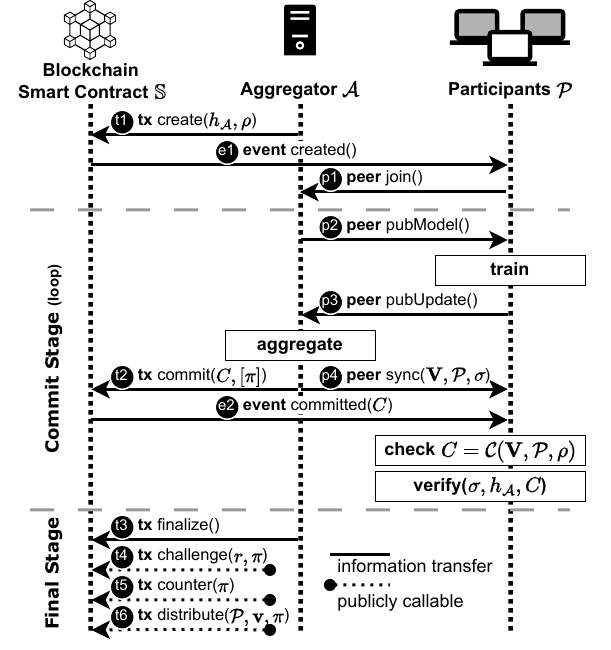}
    \caption{System Workflow.
        \textnormal{\textbf{tx}: on-chain transaction; \textbf{event}: smart contract event; \textbf{peer}: off-chain peer communication.}
    }
    \Description{System Workflow.
        ## Roles:
        - Blockchain (Smart Contract)
        - Aggregator
        - Participants

        ## Sequence:

        **Init**
        - (t1) tx create()
        - (e1) event created()
        - (p1) peer join()

        **Commit Stage (loop)**
        - (p2) peer pubModel()
        - action train()
        - (p3) peer pubUpdate()
        - action aggregate()
        - (t2) tx commit()
        - (p4) peer sync()
        - (e2) event committed()
        - action check()
        - action verify()

        **Final Stage**
        - (t3) tx finalize()
        - (t4) tx challenge()
        - (t5) tx counter()
        - (t6) tx distribute()
    }
    \label{fig:workflow}
\end{figure}

\paragraph{Workflow}

The \zkrfl{} and \zkccfl{} protocols share a unified workflow sequence across multiple FL rounds, as illustrated in Figure~\ref{fig:workflow}.
The \agg \step{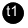} creates a task by staking rewards and publishing metadata on-chain.
Upon the \step{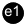} event, one of \pars, denoted as $P_i$, \step{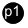} joins a task of interest via the peer-to-peer network.
Each round begins when the \agg \step{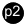} broadcasts the global model.
$P_i$ trains locally and \step{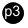} publishes an update to the \agg.
After collecting updates, the \agg aggregates them and \step{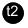} commits a compressed state to the blockchain.
For \zkrfl, this commitment acts as a validity proof, whereas \zkccfl requires only the state data.
The \agg also \step{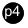} synchronizes the latest state and signature with $P_i$.
Upon the \step{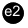} commitment event, $P_i$ validates its update via $\mathcal{C}$ and signature verification, then decides whether to continue or drop out.
After the final round, the \agg \step{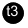} marks the task as complete.
In the \zkccfl workflow, if any inconsistency is found, $P_i$ or any other user may \step{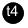} challenge the commitment with a SNARK proof.
The \agg must then \step{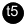} respond with a counter-proof before the dispute window closes.
Finally, the \agg or any user may \step{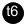} trigger reward distribution.
This action is available immediately for \zkrfl, while \zkccfl requires the dispute window to expire without valid challenges.

\subsection{Commitment Scheme}
\label{sec:commitment-scheme}
\begin{figure*}[ht]
    \centering
    \includegraphics[width=\textwidth]{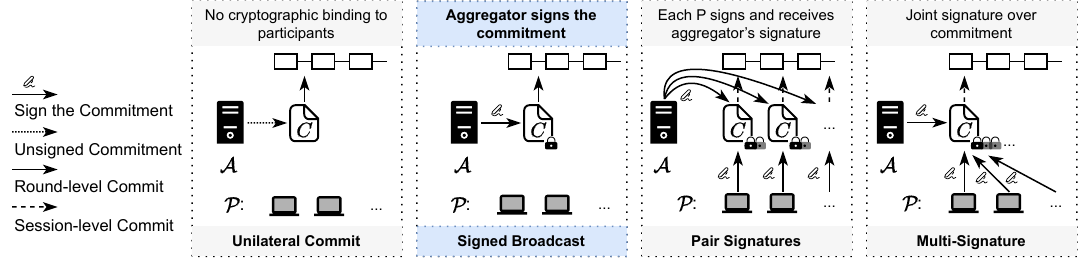}
    \caption{Design space of cryptographic binding mechanisms for reward enforcement.
        \textnormal{We show four approaches with increasing stronger binding. We adopts the second, using \emph{aggregator-signed commitments} to balance cost, scalability, and verifiability.}
    }

    \Description{
        The diagram displays four distinct commitment architectures involving an Aggregator ($\mathcal{A}$), Client ($\mathcal{C}$), and Participants ($\mathcal{P}$).

        Legend key:
        - Solid arrow: Round-level Commit
        - Dashed arrow: Session-level Commit
        - Dotted arrow: Unsigned Commitment
        - Pen icon: Sign the Commitment

        1. Unilateral Commit:
        - Feature: "No cryptographic binding to participants"
        - Structure: $\mathcal{A}$ sends to $\mathcal{C}$.

        2. Signed Broadcast: (Highlighted)
        - Feature: "Aggregator signs the commitment"
        - Structure: $\mathcal{A}$ signs and transmits to both $\mathcal{C}$ and $\mathcal{P}$.

        3. Pair Signatures:
        - Feature: "Each $\mathcal{P}$ signs and receives aggregator's signature"
        - Structure: $\mathcal{A}$ engages in bidirectional signing with multiple $\mathcal{P}$s.

        4. Multi-Signature:
        - Feature: "Joint signature over commitment"
        - Structure: $\mathcal{A}$ and multiple $\mathcal{P}$s coordinate for a collective signature.
    }

    \label{fig:design-space}
\end{figure*}

The foundation of our protocol family centers on the unified notion of a \emph{commitment}: one can prove knowledge of its content or the correctness of its transition, which together underpin reward enforcement in both validity and optimistic workflows.

\paragraph{Rationale.}
Blockchain offers a natural foundation for trustless reward settlement in FL.
However, on-chain execution remains costly due to gas cost and contract size limits (e.g., 24KB on Ethereum), creating scalability challenges for large participant sets.
While alternative chain solutions reduce these costs, they often trade off trust assumptions.

We instead explore protocol-level strategies that minimize interaction while preserving decentralization.
As illustrated in Figure~\ref{fig:design-space}, binding commitments to participants can be realized through several approaches, ranging from unilateral commits with no authentication, to fully aggregated multi-signatures.
While stronger cryptographic binding (e.g., multi-signature) can achieve immediate finality, it often incurs high latency and coordination overhead, limiting scalability.

Our protocol therefore adopts the \textbf{aggregator-signed broadcast model}, where the aggregator signs each round's commitment and participants verify it off-chain.
This choice avoids costly multi-party coordination while still providing strong attribution, but accepts a bounded risk: if the aggregator misbehaves, at most one round of rewards may be lost before participants can exit.
This trade-off enables light on-chain logic, low latency in optimistic execution, and scalable support for large FL jobs.

\paragraph{Cryptographic Primitives.}
In the circuit, we use Poseidon~\cite{grassiPoseidon2021}, a SNARK-friendly hash function, with notation $H_k(\cdot)$ for Poseidon of arity $k$.
For signatures, we adopt EdDSA~\cite{bernsteinHighSpeed2011}; since an EdDSA public key consists of two curve coordinates, we use its hash $\hagg = H(\mathbf{A}_x,\mathbf{A}_y)$ as the aggregator's on-chain identifier to minimize public signals and gas cost.
In the Solidity-based \ccfl implementation, Keccak and ECDSA are used instead, and we write $H(\cdot)$ and $\operatorname{sign}(\cdot)$ for generic hash and signature.

\paragraph{Notation.}
We use $[T]$ to denote the set $\{0,1,\dots,T-1\}$.
$\mathbf{1}[\cdot]$ is the indicator function, which equals $1$ if the condition is true and $0$ otherwise.
We define four constants: $T$ is the number of rounds, $N$ is the number of participants, $B$ is the batch size, and $\kappa$ is the arity of the Poseidon hash function.

\paragraph{Commitment.}
An aggregator \agg is required to publish on-chain its public key hash $\hagg$, which uniquely binds its identity.
A set of participants $\pars = \{P_1, \ldots, P_N\}$ collaboratively train a global model, where each $P_i$ is identified by its address.
In each round $r$, \agg commits a reward allocation vector $\mathbf{v}^r = (v_1^r, \ldots, v_N^r)$, where $v_i^r$ is the reward assigned to $P_i$.

A commitment $C$ is defined as
\[
    C = H(H(\bm{V}), H(\pars), \salt),
\]
where $\bm{V}=\{\mathbf{v}^0,\dots,\mathbf{v}^{T-1}\}$ is the reward allocation matrix, $H$ is the hash function, and $\salt$ is a random salt.
If the effective number of rounds or participants is less than the predefined maximum $T$ and $N$, $\bm{V}$ is padded with $0$ since the circuit size is fixed at compile time.
\agg signs $C$ with its private key, producing $\sigma_C = \operatorname{sign}_{\agg}(C)$.
Each round, $C$ is posted on-chain, while $\sigma_C$ is broadcast to participants.
Participants verify $\sigma_C$ against $\pubkeyA$ to ensure authenticity.

If anyone detects inconsistency, they can issue a challenge by providing a SNARK proof that references the last valid commitment at round $r-1$.
To protect \agg from spurious challenges, it may counter by presenting a SNARK proof of a valid transition from $C^{r-1}$ to $C^r$.

\paragraph{Implementation.}
Because the protocol family shares the core verification logic, implementing these shared components natively instantiates both workflows.
The logic is written directly in Solidity for \ccfl, and as Groth16 SNARK circuits for both \zkccfl and \zkrfl.
In Solidity, no padding is required since inputs are handled natively by the EVM.

\subsection{Smart Contract}
\label{sec:contract}

The smart contract serves as the core of reward enforcement in \sys{}, implemented through a state-driven design that we describe in terms of its state machine and contract logic.

\paragraph{State Machine.}

As shown in Figure~\ref{fig:contract_flow}, the contract implements a state machine coordinating the job lifecycle: commitment, reward initialization, optional challenge/counter, and final settlement through one-shot (1s) or multi-shot (ms) distribution.
Most transitions are externally triggered, except for timeout-based transitions (e.g., delayed entry into the distribution phase).

Since only the commitment is on-chain, proofs are required for committing, challenging, or distributing.
In both \zkrfl and \zkccfl, a SNARK proof is generated by a SNARK prover; in \ccfl, the proof consists of all original data used to compute the commitment.
\begin{figure}[t]
    \centering
    \includegraphics[width=\columnwidth]{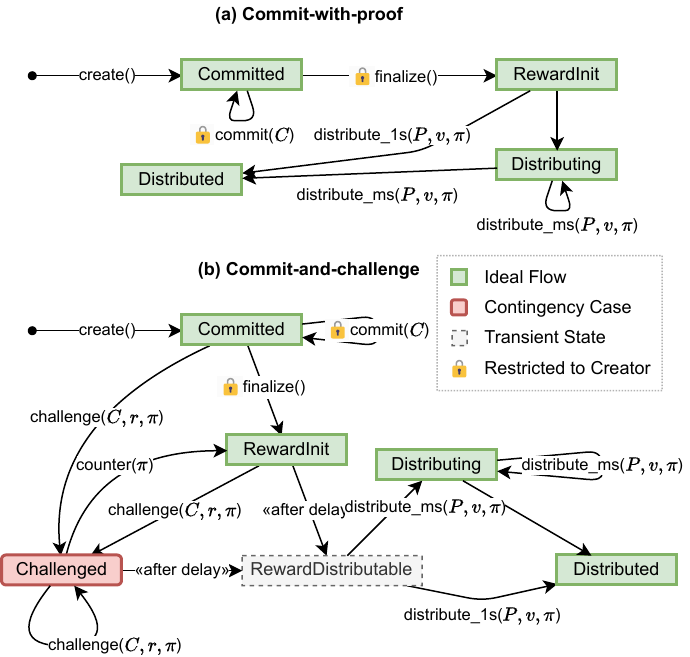}
    \caption{
        Contract state transitions.
        \textnormal{In (b) Commit-and-challenge, if a dispute exists, challenges are allowed at any point before distribution. The protocol follows either a multi-shot (ms) or one-shot (1s) path based on the chosen distribution method.}
    }
    \Description{
        The figure displays two related state transition diagrams labeled (a) and (b).

        ## (a) Commit-with-proof
        [*] --> A_Committed: create()
        A_Committed --> A_Committed: commit(C)
        A_Committed --> A_RewardInit: finalize()

        A_RewardInit --> A_Distributing
        A_RewardInit --> A_Distributed: distribute_1s(P, v, \pi)

        A_Distributing --> A_Distributing: distribute_ms(P, v, \pi)
        A_Distributing --> A_Distributed: distribute_ms(P, v, \pi)

        ## (b) Commit-and-challenge
        [*] --> B_Committed: create()
        B_Committed --> B_Committed: commit(C)
        B_Committed --> B_RewardInit: finalize()

        B_Committed --> Challenged: challenge(C, r, \pi)
        B_RewardInit --> Challenged: challenge(C, r, \pi)
        Challenged --> Challenged: challenge(C, r, \pi)
        Challenged --> B_RewardInit: counter(\pi)
        Challenged --> RewardDistributable: after delay

        B_RewardInit --> RewardDistributable: after delay

        RewardDistributable --> B_Distributing
        RewardDistributable --> B_Distributed: distribute_1s(P, v, \pi)

        B_Distributing --> B_Distributing: distribute_ms(P, v, \pi)
        B_Distributing --> B_Distributed: distribute_ms(P, v, \pi)

    }

    \label{fig:contract_flow}
\end{figure}

\paragraph{Smart Contract Logic.}

Algorithm~\ref{alg:ccfl_core} defines the logic of $\mathbb{S}$.
It maintains the \agg's public key hash $\hagg$, current round $r$, salt $\salt$, unlock time $\unlockTime$, distribution count $\xi$, verifier function $\verify{(\cdot)}$, commitment history $\mathbf{C} = [C^1, \ldots, C^{r-1}]$, distribution mode $\mu \in \{\textit{one-shot},\textit{multi-shot}\}$, current status $\status \in \{\statusCommitted, \statusRewardInit,\allowbreak \statusChallenged, \statusPaying,\allowbreak \statusPaid\}$, the last successfully challenged round $\lastround$, and the protocol variant $\mode \in \{\modeChallenge,\modeProof\}$.
In $\modeProof$, $\opcommit$ requires a succinct proof $\proofpi{commit}$ to be verified.
$\opchallenge$ and $\opcounter$ are enabled in $\modeChallenge$.
Functions with $\onlyOwner$ are owner restricted.
\begin{algorithm}[t]
    \caption{Smart Contract $\mathbb{S}$ for $\Pi_{\piccfl}$}\label{alg:ccfl_core}
    \small

    \SetKwProg{func}{function}{}{}

    \func{$constructor$($\hagg$, $\salt$, $\mu$, $\mode$) \textbf{public}} {
        $\mathbf{C} \leftarrow [~]$, $r \leftarrow 1$, $r_{\text{last}} \leftarrow 2$,
        $\status \leftarrow \statusCommitted$; \\
    }

    \func{$\opcommit$($C$, $[\proofpi{commit}]$) \textbf{public} \onlyOwner}  {
        \textbf{assert} $\status \in \{\statusCommitted\}$; \\
        \If{$\mode = \modeProof$}{
            \textbf{assert} $\verify{transition}$($r, \mathbf{C}^{r-1}, C, \hagg; \proofpi{commit}$); \\
        }
        append $C$ to $\mathbf{C}$; \\
        $r \leftarrow r + 1$;
    }

    \func{$\opfinalize$() \textbf{public} \onlyOwner} {
        \textbf{assert} $\status \in \{\statusCommitted, \statusChallenged\}$; \\
        $\status \leftarrow \statusRewardInit$,
        $\unlockTime \leftarrow \textit{block.timestamp} + \Delta$;
    }

    \func{$\opchallenge$($r'$, $\proofpi{challenge}$) \textbf{public}} {
        \textbf{assert} $\mode = \modeChallenge$; \\
        \textbf{assert} $\status \in \{\statusCommitted, \statusRewardInit, \statusChallenged\}$; \\
        \textbf{assert} $r' > \lastround$; \\
        \textbf{assert} $\verify{challenge}$($r', \mathbf{C}^{r'-2}, \hagg; \proofpi{challenge}$); \\
        $\lastround \leftarrow r'$ \Comment*[r]{rounds $< r'$ are stale}
        $\status \leftarrow \statusChallenged$,
        $\unlockTime \leftarrow \textit{block.timestamp} + \Delta$;
    }

    \func{$\opcounter$($r'$, $\proofpi{counter}$) \textbf{public}} {
        \textbf{assert} $\mode = \modeChallenge$; \\
        \textbf{assert} $\status \in \{\statusChallenged\}$; \\
        \textbf{assert} $r' \geq 3$; \\
        \textbf{assert} $\verify{transition}$($r'-1, \mathbf{C}^{r'-2}, \mathbf{C}^{r'-1}, \hagg; \proofpi{counter}$); \\
        $\status \leftarrow \statusRewardInit$,
        $\unlockTime \leftarrow \textit{block.timestamp} + \Delta$;
    }

    \func{$\opdistribute$($\vec{p}$, $\vec{s}$, $\proofpi{distribute}$) \textbf{public}} {
        \textbf{assert} $\status \in \{\statusRewardInit, \statusChallenged\}$; \\
        \textbf{assert} \textit{block.timestamp} $\geq \unlockTime$; \\

        \textbf{assert} $\verify{distribute}$($r-1, \mathbf{C}^{r-1}, \hagg, \vec{s}, \vec{p}; \proofpi{distribute}$); \\

        \For{$i = 1$ \KwTo $|\vec{p}|$}{
            \textit{transfer}($p_i$, $s_i$);
        }

        \If{$\mu = \text{multi-shot}$}{
            $\xi \leftarrow \xi + B$\;
            $\status \leftarrow \statusPaid \;\mathbf{if}\; \xi \geq |\pars| \;\mathbf{else}\; \statusPaying$\;
        } \Else{
            $\status \leftarrow \statusPaid$ \Comment*[r]{In one-shot, $\vec{p}=\pars$}
        }
    }
\end{algorithm}

\subsection{Circuit Gadgets}
\label{sec:gadgets}

We present a set of reusable gadgets that serve as building blocks for the verifier circuit.
Throughout, we denote constraints by $\langle \cdot \rangle$, indicating that the enclosed relation must hold in the proof system's field.
Within a constraint, ``$=$'' expresses equality to be enforced, while ``$\gets$'' denotes assignment of a value, with the resulting equality enforced implicitly.

The detailed constraint specifications are provided in Appendix~\ref{appendix:verifier_gadgets}; here we focus on the design rationale and semantics.

\paragraph{Rationale.}
We opt for a domain-specific architecture over a generic Zero-Knowledge Virtual Machine (zkVM)~\cite{ben-sassonSNARKs2013} to avoid prohibitive computational overhead.
First, supporting generic instructions makes zkVM circuits orders of magnitude larger than our tailored design.
Second, \sys{} uses SNARK-friendly primitives like Poseidon and EdDSA.
A Poseidon hash requires only about 300 constraints, compared to 150,000 for Keccak-256~\cite{arnaucubeKeccak256circom2021}.
Third, by exploiting the fixed structure of FL reward matrices, our \gadgetname{LPI} and \gadgetname{BatchExtractor} gadgets achieve $O(T)$ round selection and $O(N)$ participant extraction with linear constraint counts.
These optimizations reduce proof generation time from minutes to seconds, making per-round commitments economically viable.

\paragraph{Length-Padded Indicator (\gadgetname{LPI}).}
$\gadgetname{LPI}(x)$ takes an integer $x$ and produces a binary vector of length $T$ with the first $x$ entries set to $1$ and the rest set to $0$.
Thus $\gadgetname{LPI}(x) = (\{1\}^x, \{0\}^{T-x})$.
This vector is mainly used as a prefix selector for rounds or batches.
The constraints enforce three properties: the sum of outputs equals $x$, each entry is boolean, and there is exactly one $1\!\to\!0$ transition.
These constraints ensure a scalable $O(T)$ construction.

\paragraph{Zero Comparator (\gadgetname{IsZero}).}
The $\gadgetname{IsZero}(x)$ gadget outputs $y\in\{0,1\}$ indicating whether $x=0$.
It uses an auxiliary witness $\mathit{inv}$, which equals $0$ if $x=0$ and $x^{-1}$ otherwise.
This trick allows checking equality to zero with only two constraints.

\paragraph{Modular Hash ($H_{M}$).}
Poseidon supports only a bounded arity up to a maximum of $16$ due to best security practices; therefore, we set $\kappa=16$ as the default.
We therefore chunk a long vector $(x_1,\dots,x_N)$ into groups of at most $\kappa-1$ elements,
and iteratively absorb them into a chained hash state seeded with $\salt$.
This \emph{modular hash} ensures that any input length can be processed, while minimizing the number of Poseidon hashes reduces per-element cost.

\paragraph{Batch Extractor (\gadgetname{BatchExtractor}).}
Given a sequence $\vec{x}$ of length $N$, a batch size $B$, and an index $a$,
$\gadgetname{BatchExtractor}(\vec{x},a) \break = (x_{a B}, \ldots, x_{a B+B - 1})$ outputs the $a$-th batch of size $B$.
To reduce constraints, we introduce an intermediate selection matrix enabled by $\gadgetname{IsZero}$, leading to linear complexity in $N$ instead of exponential naive conditional selection.

\paragraph{Vectorization ($\text{vec}$).}
For hashing purposes, a reward matrix $V\in\mathbb{F}^{N\times T}$ is serialized into a one-dimensional vector in row-major order:
$\text{vec}(V) = (V_{0,0}, V_{0,1},\dots,V_{0,T-1},\,V_{1,0},\dots,V_{N-1,T-1})$.
This ensures a canonical ordering before applying the modular hash.

\paragraph{Matrix Sum Row ($\gadgetname{SumRow}$).}
Given a LPI vector $\vec{r}$, this gadget extracts the contributions of each participant across all rounds from a reward matrix $V$ to a vector $\vec{y}$.

\paragraph{Matrix Round Mask Check ($\gadgetname{MaskCheck}$).}
This gadget ensures $V$ is nonzero only in the selected round.
We enforce \constrain{V_{i,t}\cdot (1-r_t)=0}, $\forall (i,t) \in [N]\times[T]$, to ensure consistency with the LPI vector $\vec{r}$.

\paragraph{Public Key Digest ($H_\text{PK}$).}
To compress public signals to reduce on-chain gas cost, the aggregator's digest is $\hagg = H_\text{PK}(\mathbf{A}) = H_2(\mathbf{A}_x, \mathbf{A}_y)$,
reducing two field elements to one hash output.

\paragraph{Verify Commitment ($\gadgetname{VerifyCommitment}$).}
A wrapper gadget that combines vectorization, modular hashing, and EdDSA signature verification.
It outputs the canonical commitment $C$ and checks that the aggregator's signature $\sigma_C$ on $C$ is valid under its public key.

\subsection{Main Circuits}
\label{sec:circuits}

By composing the above gadgets, we obtain a set of circuits that capture the essential actions of our protocol.
Together, they form a closed system: commitments can be published, challenged, countered, and finally distributed,
all with succinct verification and bounded circuit complexity.
For completeness, the pseudocode of these verifiers is provided in Appendix~\ref{appendix:verifier_circuits}, where $\pubpar{underlined}$ inputs or outputs are public and all others are private.

Notably, all circuits compute and output a public key aggregation hash $\hagg$ derived from $\mathbf{A}$ to bind the operations to a specific participant group , and they utilize a private $\salt$ in the commitment verification to ensure cryptographic hiding properties.
We then describe the role of each circuit in the protocol and how they jointly form a closed and verifiable system.

\paragraph{Transition Verifier.}
This circuit instantiates $\verify{transition}$ and checks that two consecutive commitments $(C_t,C_{t+1})$ are consistent.
It ensures that the reward matrix $\bm{V}_2$ for round $t+1$ is valid using a \gadgetname{MaskCheck} gadget , and that both commitments are correctly verified using their respective signatures and the shared $\salt$.
Since $C_t$ is derived directly as a slice of the later matrix, and similarly, its participant vector $\vec{p}_1$ is derived from $\vec{p}_2$ via LPI filtering, their validity does not need to be rechecked from scratch.
This verifier is used in the \emph{commit} action of \zkrfl and the \emph{counter} action of \zkccfl.

\paragraph{Challenge Verifier.}
This circuit instantiates $\verify{challenge}$ and proves knowledge of the data underlying a given commitment $C$, together with a valid signature and $\salt$.
It additionally utilizes \gadgetname{MaskCheck} to ensure the matrix $\bm{V}$ accurately belongs to the claimed round.
It does not by itself prove misbehavior; rather, it ensures that a challenger can faithfully reconstruct the disputed state.
Whether the commitment is actually inconsistent is determined only when a \emph{counter} action provides a transition proof.

\paragraph{Distribution Verifier.}
This circuit instantiates $\verify{distribute}$ and besides commitment verification, computes the participant-level reward vector for a specific batch $b$.
It sums the reward matrix along the selected round into a vector $\vec{s}$ , and utilizes a \gadgetname{BatchExtractor} to extract the relevant subset of participants $\vec{p}_p$ and their corresponding rewards $\vec{s}_p$.
These public outputs can be directly used to settle payments in the \emph{distribute} action.

\subsection{Incentive Compatibility}
\label{sec:incentive}

While \zkrfl{} relies on cryptographic soundness, \zkccfl{} achieves rational security through economic incentives.
To guarantee rational security in the dispute resolution process, we first formalize the concept of \emph{incentive compatibility (IC)}.

\begin{definition}[Incentive Compatibility]
    A protocol is \emph{incentive compatible} if for every player $i \in M$, assuming all other players follow the protocol, the honest strategy $s_i^*$ maximizes player $i$'s utility:
    $$
    u_i(s_i^*, s_{-i}^*) \geq u_i(s_i', s_{-i}^*) \quad \forall s_i' \in \mathcal{S}_i\,,
    $$
    where $\mathcal{S}_i$ is the strategy space of player $i$, and $u_i(s_i, s_{-i})$ denotes $i$'s utility under strategy profile $(s_i, s_{-i})$.
\end{definition}

\paragraph{Game Model.}
To evaluate IC, we model the interaction as a static complete-information normal-form game~\cite{costa-gomesCognition2001}:
$$
\mathcal{G} = (M, \{\mathcal{S}_i\}_{i \in M}, \{u_i\}_{i \in M}).
$$
The player set is $M = \{\agg, P_1, \dots, P_N\}$.
The aggregator's strategy space $\mathcal{S}_{\agg}$ includes honest execution, tampering with rewards, or aborting.
A participant's strategy space $\mathcal{S}_{P_i}$ includes honestly verifying and challenging when justified, submitting malicious challenges, or remaining passive.

Under this formulation, we establish that mutual honest behavior constitutes a Nash Equilibrium~\cite{myersonValue1992}.

\begin{theorem}[Nash Equilibrium under Rationality]
    If all players are rational and protocol parameters are set properly, then the strategy profile $(s_{\agg}^{\text{honest}}, s_{P_1}^{\text{honest}}, \dots, s_{P_N}^{\text{honest}})$ forms a Nash equilibrium of $\mathcal{G}$.
\end{theorem}

In a real deployment, a player's utility $u_i$ is determined by a complex set of parameters, including model rewards, task bonuses, slashing penalties, and on-chain execution costs.
However, as detailed in our formal game-theoretic analysis in Appendix~\ref{appendix:incentive}, ensuring the equilibrium mathematically reduces to a single fundamental condition: $P_{\agg}^{\text{slash}} > C_i^{\text{gas}}$.

Intuitively, the aggregator has no incentive to tamper because dishonest commitments will be challenged and penalized.
For a participant, submitting a malicious challenge is irrational as it will be countered and penalized.
Therefore, the only rational deviation for a participant is to remain passive to save verification costs.
To ensure honest verification strictly dominates passive behavior, the potential slashing bounty $P_{\agg}^{\text{slash}}$ awarded for catching a malicious aggregator must strictly outweigh the gas cost $C_i^{\text{gas}}$ required to submit the challenge on-chain.

\section{Experimental Evaluation}
\label{sec:experiment}
\begin{table}[t]
    \caption{Characteristics of datasets.}
    \label{tab:fl-datasets}
    \centering
    \small
    \setlength{\tabcolsep}{3pt}
    \begin{tabular}{l|ccc}
        \toprule
        \textbf{Dataset} & \textbf{Samples} & \textbf{Classes} & \textbf{Task} \\
        \midrule
        MNIST~\cite{dengMNIST2012}& 70,000 & 10 & image classification  \\
        Fashion-MNIST~\cite{xiaoFashionMNIST2017} & 70,000 & 10 & image classification \\
        CIFAR-10~\cite{krizhevskyLearning2009} & 60,000 & 10 & image classification \\
        Bank~\cite{moroDatadriven2014} & 45,211 & 2 & tabular data classification \\
        Cora~\cite{mccallumAutomating2000} & 2,708 & 7 & graph node classification \\
        Twitter~\cite{goTwitter2009} & 1,194,723 & 2 & text classification \\
        \bottomrule
    \end{tabular}
\end{table}

We evaluate \sys{} through experiments coupled with both local and permissionless blockchain deployments.
Both \zkccfl{} and \zkrfl{} are implemented as distinct modes within a single smart contract.
Furthermore, for brevity in the figures throughout this section, we use the abbreviations CC and CP to denote \zkccfl{} and \zkrfl{}, respectively.

FL experiments primarily use the MNIST dataset~\cite{dengMNIST2012}, which serves as the main benchmark for end-to-end evaluation.
Additional experiments are included to assess generality across all datasets listed in Table~\ref{tab:fl-datasets}.
Each dataset follows its canonical train/test split; Twitter follows LEAF~\cite{caldasLEAF2019}.
We use LeNet-5~\cite{lecunGradientbased1998} for MNIST, a lightweight 4-layer CNN~\cite{goodfellowDeep2016} for Fashion-MNIST, ResNet-20~\cite{heDeep2016} for CIFAR-10, a 3-layer MLP~\cite{goodfellowDeep2016} for Bank, a 2-layer GCN~\cite{kipfSemiSupervised2017} for Cora, and FastText~\cite{joulinBag2017} for Twitter.

Blockchain experiments are conducted on a local Ethereum-compatible testnet (Hardhat), while end-to-end validation is additionally deployed on the public Ethereum testnet (Sepolia).
SNARK circuits are written in Circom~\cite{belles-munozCircom2023} and verified using Groth16 via snarkjs~\cite{baylinaSnarkjs2018}.
The simulation framework is built with Node.js, incorporating TensorFlow.js~\cite{smilkovTensorFlowjs2019} for local training and ethers.js for contract interaction.
All experiments are conducted on a server with dual Intel Xeon CPUs, 512 GB RAM and eight NVIDIA RTX 4090 GPUs, running Ubuntu 22.04.

\subsection{End-to-End Protocol Validation}

We validate the end-to-end functionality of \sys{} through real blockchain deployments on the public Ethereum Sepolia network, which provides realistic gas costs, latency, and contract behavior.
This setup closely approximates real-world conditions while isolating protocol correctness and scalability from infrastructural variability.

All \pars join the task, receive the global model, perform local training, and submit updates.
Depending on the chosen variant, the \agg either commits a reward allocation subject to challenge in \zkccfl{} or commits a validity proof for immediate settlement in \zkrfl{}.
\begin{figure}[t]
    \centering
    \includegraphics[width=0.47\textwidth, trim=8 10 4 10]{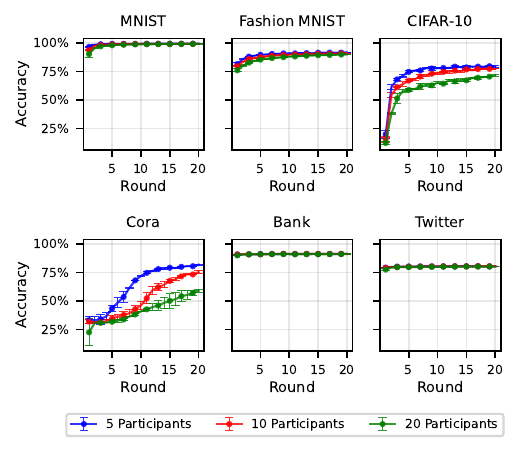}
    \caption{
        End-to-end protocol accuracy over training rounds under different participant counts.
        \textnormal{
            Results use a default learning rate of 0.001, a batch size of 32, and 2 local epochs per round, except for Cora (learning rate 0.015, 10 local epochs) and Twitter (1 local epoch).
        }
    }
    \Description{
        Three-panel line chart: the left, middle, and right panels correspond to MNIST, Fashion-MNIST, and CIFAR-10, respectively.
        In each panel, three curves represent experiments with 5, 10, and 20 participants.
        The x-axis denotes training rounds (1--20), and the y-axis shows global model test accuracy in percent.
        Vertical error bars indicate the minimum and maximum accuracies across multiple independent runs evaluated by the aggregator.
        All panels present results for exactly 20 training rounds under identical hyperparameters.
    }
    \label{fig:e2e-20-rounds-different-participants-accuracy}
\end{figure}

\paragraph{Training Performance.}
Participants are emulated on a single server for controllability, performing local training and aggregation as in standard FL.
The aggregator commits round-level reward allocations on-chain and broadcasts updated global models off-chain.
We evaluate accuracy progression across 20 training rounds under varying participant counts (5--20).
Figure~\ref{fig:e2e-20-rounds-different-participants-accuracy} reports results on MNIST, Fashion-MNIST, and CIFAR-10.
Both protocol variants share the same off-chain training logic, achieving accuracy comparable to standard FL.
This confirms that neither the reward enforcement mechanism of \zkccfl{} nor the proof generation of \zkrfl{} interferes with the model accuracy.

\paragraph{Realistic Deployment.}

To evaluate the protocol's scalability and practicality under realistic blockchain conditions, we deployed \sys{} on the Sepolia testnet.
Specifically, both \zkccfl{} and \zkrfl{} are evaluated on the MNIST dataset with 800 \pars{} over 50 training rounds.
The representative \zkccfl{} deployment (\texttt{0xca7b...}\footnote{\url{https://sepolia.etherscan.io/address/0xca7b84856557a065C82f6425E374f2912094EDA6}}) achieves a final accuracy of 94.55\%.
The representative \zkrfl{} deployment (\texttt{0x124c...}\footnote{\url{https://sepolia.etherscan.io/address/0x124c8F01728ea88b19B9CFd3B4a5981858A425F8}}) successfully enforces strict cryptographic verification while maintaining equivalent model convergence.
\begin{table}[t]
    \centering
    \caption{Cost breakdown comparison of representative Sepolia deployments for \sys{}. (800 \pars, 50 rounds on MNIST) }
    \label{tab:sepolia-cost}
    \begin{tabular}{lrrrr}
        \toprule
        & \multicolumn{2}{c}{\zkrfl} & \multicolumn{2}{c}{\zkccfl} \\
        \cmidrule(lr){2-3} \cmidrule(lr){4-5}
        \textbf{Operation} &
        \multicolumn{1}{c}{\textbf{mETH}} & \multicolumn{1}{c}{\textbf{USD}} &
        \multicolumn{1}{c}{\textbf{mETH}} & \multicolumn{1}{c}{\textbf{USD}} \\
        \midrule
        \opjobcreation & 0.012 & 0.02 & 0.012 & 0.02 \\
        \opcommit\ (per call) & 0.021 & 0.04 & 0.005 & 0.01 \\
        \opcommit\ (50$\times$ total) & 1.066 & 2.08 & 0.239 & 0.47 \\
        \opfinalize & 0.005 & 0.01 & 0.005 & 0.01 \\
        \opchallenge & -- & -- & 0.021 & 0.04 \\
        \opcounter & -- & -- & 0.020 & 0.04 \\
        \opdistributem\ (per call) & 0.110 & 0.21 & 0.110 & 0.21 \\
        \opdistributem\ (16$\times$ total) & 1.757 & 3.43 & 1.757 & 3.43 \\
        \midrule
        \textbf{Total} & \textbf{2.841} & \textbf{5.54} & \textbf{2.054} & \textbf{4.00} \\
        \bottomrule
    \end{tabular}

    \vspace{0.5em}

    \begin{minipage}{0.95\linewidth}
        \small
        \textbf{Gas price:} 0.075 gwei (ETH 7-day median, Feb 14--20, 2026).\\
        \textbf{ETH/USD:} 1,950 (snapshot on Feb 20, 2026).\\
        \textbf{Note:}
        Costs are in USD and mETH (1 mETH = 0.001 ETH).
    \end{minipage}
\end{table}

Table~\ref{tab:sepolia-cost} summarizes the gas cost breakdown for both variants.
These deployments incur total costs of 0.0028 ETH for \zkrfl{} and 0.0021 ETH for \zkccfl{}, corresponding to approximately \$5.54 and \$4.00, respectively.
As theoretically analyzed in previous sections, the fundamental divergence between the two variants lies in the \opcommit{} operation.
For \zkccfl{}, each \opcommit{} incurs minimal gas overhead, with the trade-off being the potential necessity of \opchallenge{} and \opcounter{} operations to resolve disputes.
Conversely, \zkrfl{} achieves immediate finality upon submission.
It securely validates the update via proof verification, entirely eliminating the need for interactive dispute resolution.
These empirical results directly validate our earlier claims regarding the trade-offs.
Furthermore, they provide a clear guideline for future deployment: \zkccfl{} is preferable for scenarios prioritizing minimal routine on-chain execution costs, whereas \zkrfl{} is better suited for applications demanding strict cryptographic certainty and non-interactive finality.

\subsection{Robustness Against Misbehavior}

We evaluate the robustness of \zkccfl{} under adversarial behavior from both aggregators and challengers.
In the first scenario, a malicious aggregator commits an incorrect model update.
Once challenged, the aggregator fails to produce a valid counter-proof and is thereby disqualified.
The challenger incurs a one-time cost of approximately 27.5k gas to initiate the challenge.

In the second scenario, a challenger repeatedly attempts to disrupt the protocol by submitting stale proofs from earlier rounds.
The initial challenge consumes about 27.5k gas, while each subsequent attempt costs roughly 25.7k gas.
The aggregator, upon issuing a valid counter, spends around 26k gas per response.
A successful counter leads to the forfeiture of the challenger's bond, which is transferred to the aggregator as compensation.
These results demonstrate that the protocol effectively deters unproductive disputes while preserving correctness through verifiable commitments.

\subsection{Cost and Latency Analysis}
\label{sec:cost_latency_analysis}

We compare four protocol variants: \zkccfl, \ccfl, \zkrfl, and \bcfl.
Figure~\ref{fig:gas-latency-comparison} reports the average gas cost and latency per round.
Specifically, the gas cost is computed over experiments with 10, 50, 100, and 200 training rounds, and the latency is evaluated under the 10 and 50 round settings.
Proofs, when required, were generated using \texttt{rapidsnark}~\cite{lomakaRapidsnark2025}.
\begin{figure}[t]
    \centering
    \includegraphics[width=0.47\textwidth, trim=3 10 0 10]{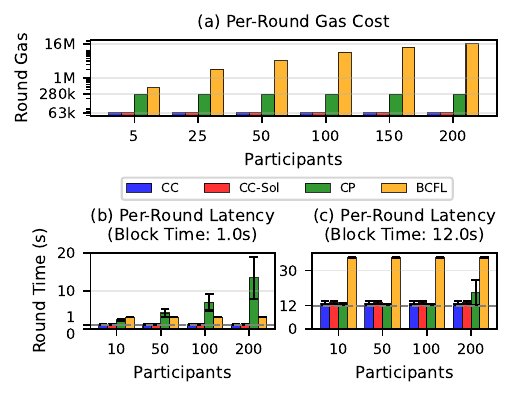}
    \caption{Per-round cost and latency comparison.
        \textnormal{\sys{} variants achieve stable gas usage and low latency.}
    }
    \Description{}
    \label{fig:gas-latency-comparison}
\end{figure}

As shown in Figure~\ref{fig:gas-latency-comparison}(a), both \zkccfl and \ccfl maintain stable gas cost across rounds and participants.
\zkrfl has constant but higher cost due to per-round proofs.
\bcfl scales linearly and soon becomes impractical.
Latency is estimated under two block times: 1s to simulate fast chains, and 12s to reflect Ethereum's average.
\zkccfl and \ccfl stay efficient in both cases.
\zkrfl suffers from proof delays, and its latency grows with the number of rounds, leading to larger error bars.
\bcfl accumulates on-chain bottlenecks.

Results confirm that the \sys{} family achieves low and predictable cost and latency, outperforming other designs.

\paragraph{Comparison to Prior Works.}

OpenFL~\cite{wahrstatterOpenFL2024} extends \bcfl{} with a feedback mechanism, but its baseline gas cost remains in the same order of magnitude as \bcfl{}.
StateFL~\cite{iacobanStateFL2024} instead relies on Perun~\cite{dziembowskiPerun2017} state channels.
Since its original paper did not report consistent gas values, we re-implemented its workflow using the same stack and measured its per-participant cost.
StateFL is round-invariant but scales linearly in participants, whereas \sys{} is participant-invariant but grows with rounds.
Figure~\ref{fig:statefl-incc-comparison} summarizes the trade-off.
\begin{figure}[t]
    \centering
    \includegraphics[width=0.48\textwidth, trim=9 10 6 10]{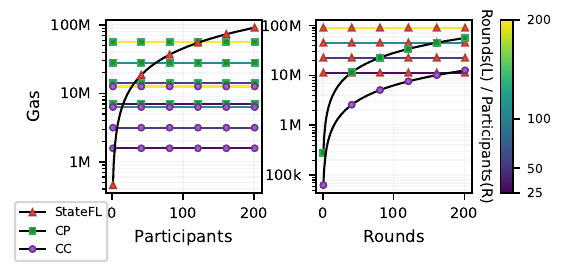}
    \caption{Gas cost comparison of StateFL and \sys{}.
        \textnormal{
            StateFL incurs 456k gas per participant. Conversely, \zkccfl{} and \zkrfl{} cost 63k and 281k gas per round, remaining more efficient except in high-round, low-participant settings.
        }
    }
    \Description{
        The figure compares gas consumption of StateFL and \sys{} across two scalability dimensions.
        \textbf{Left (x-axis: Participants):} StateFL shows a linearly increasing curve with participant count and remains flat across rounds.
        In contrast, \sys{} appears as multiple horizontal lines, each representing a fixed number of rounds (10, 50, 100, 200), since its cost does not depend on participants.
        \textbf{Right (x-axis: Rounds):} \sys{} shows a linearly increasing curve with rounds and remains flat across participants.
        StateFL appears as multiple horizontal lines, each corresponding to a fixed number of participants (10, 50, 100, 200).
        Gas values are derived from theoretical models calibrated by empirical measurements: about 456k gas per participant for StateFL (from our re-implementation on Perun) and about 63k gas per round for \sys{} (measured on Ethereum-compatible deployment).
        The curves highlight that \sys{} maintains lower cost in typical regimes and only exceeds StateFL when the number of rounds is very large (e.g., above 200) and the participant set is very small (e.g., fewer than 25).
    }
    \label{fig:statefl-incc-comparison}
    \vspace{-10pt}
\end{figure}

\subsection{Ablation Study}

To evaluate the impact of SNARK verification, we compare two versions of our protocol: \zkccfl, which uses SNARKs to offload computation off-chain, and \ccfl, a non-SNARK variant that performs all logic directly in the smart contract.

While both versions implement the same commit semantics, they differ in cryptographic primitives and execution models.
\zkccfl relies on SNARK-friendly primitives (Poseidon, EdDSA) and constant-size proofs, whereas \ccfl uses Ethereum-native operations (Keccak, ECDSA) with variable cost based on input size.

Figure~\ref{fig:challenge-counter-gas-comparison} shows the gas cost comparison across challenge and counter phases.
We observe that when SNARKs are used, the gas cost remains constant regardless of the number of participants or rounds.
In contrast, the non-SNARK implementation incurs a gas cost that grows linearly with the number of verifications.
Due to Ethereum's per-transaction gas limit ($2^{24}$), extremely large \ccfl configurations exceed the threshold and are omitted from the plot.
\begin{figure}[t]
    \centering
    \includegraphics[width=0.48\textwidth, trim=18 9 9 9]{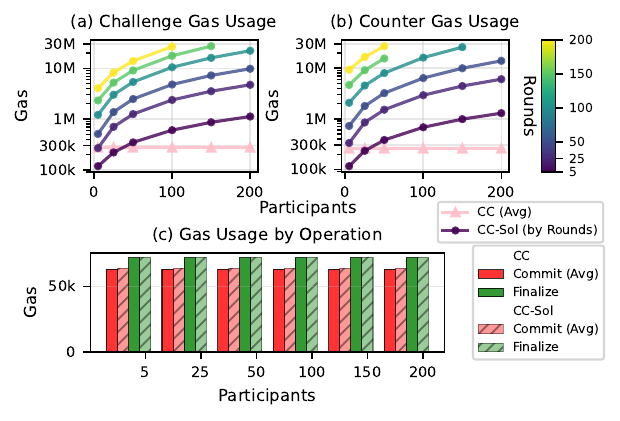}
    \caption{Gas cost of \zkccfl and \ccfl.}
    \Description{}
    \label{fig:challenge-counter-gas-comparison}
    \vspace{-10pt}
\end{figure}

Although \zkccfl incurs a higher base cost due to proof verification, it outperforms \ccfl in most practical settings, especially as the number of participants or training rounds increases.
The non-SNARK variant offers only marginal gas savings in small-scale deployments, such as settings with fewer than 25 participants and fewer than 5 rounds.
In contrast, the gas cost of \ccfl is dynamic and depends on the input size, whereas \zkccfl maintains a constant cost thanks to the fixed overhead of the Groth16 verifier.

These results confirm that SNARK-based offloading yields predictable and scalable cost profiles, making it a more robust choice for large-scale FL deployments.
In addition, the use of SNARK preserves privacy by avoiding on-chain exposure of detailed intermediate states or participant contributions which may benefit future work.
\begin{table*}[t]
    \caption{
        Constraint count and average proof generation time for main circuits across different backends and configurations.
    }
    \label{tab:zk-proof-benchmark}

    \centering
    \small
    \setlength{\tabcolsep}{1mm}
    \begin{tabular}{l|c|c|ccc|c|c|ccc|c|c|ccc}
        \toprule
        & \multicolumn{5}{c|}{\textbf{50 rounds}}
        & \multicolumn{5}{c|}{\textbf{100 rounds}}
        & \multicolumn{5}{c}{\textbf{200 rounds}} \\
        \cmidrule(lr){2-6} \cmidrule(lr){7-11} \cmidrule(lr){12-16}
        \textbf{Circuit} & \textbf{Cnstr.} & \textbf{I-GPU} & \textbf{I-CPU} & \textbf{S-CPU} & \textbf{R-CPU} & \textbf{Cnstr.} & \textbf{I-GPU} & \textbf{I-CPU} & \textbf{S-CPU} & \textbf{R-CPU} & \textbf{Cnstr.} & \textbf{I-GPU} & \textbf{I-CPU} & \textbf{S-CPU} & \textbf{R-CPU} \\
        \midrule
        \opchallenge    & 1.1/1.8 & 0.6/1.1 & 26/41 & 46/79 & 6.5/11 & 2.2/3.5 & 1.2/2.2 & 46/77 & 86/147 & 12/21 & 4.3/6.9 & 2.4/11 & 84/142 & 165/282 & 24/41 \\
        \optransition      & 2.2/3.4 & 1.2/2.2 & 47/79 & 86/146 & 12/21 & 4.3/6.8 & 2.4/11 & 84/142 & 163/286 & 23/40 & 8.4/13.5 & 12/AF & 161/281 & 314/585 & 43/83 \\
        \opdistributeo  & 1.1/1.8 & 0.6/1.1 & 26/41 & 46/78 & 6.6/11 & 2.2/3.5 & 1.2/2.2 & 44/75 & 85/147 & 12/21 & 4.3/6.9 & 2.4/11 & 85/141 & 171/288 & 23/41 \\
        \opdistributem  & 1.1/1.8 & 0.6/1.1 & 26/39 & 46/78 & 6.7/12 & 2.2/3.5 & 1.2/2.1 & 46/75 & 85/148 & 13/22 & 4.3/6.9 & 2.4/11 & 86/144 & 164/281 & 24/41 \\
        \bottomrule

    \end{tabular}

    \vspace{2pt}
    \begin{minipage}{0.95\linewidth}
        \small
        \emph{Note.} Runtime results (\textbf{I-GPU}, \textbf{I-CPU}, \textbf{S-CPU}, \textbf{R-CPU}) are reported in seconds, with each cell showing two values corresponding to runs with 500 and 800 \pars, respectively.
        Results are averaged over 10 runs.
        The \textbf{Cnstr.} column reports the number of non-linear constraints (in millions).
        \\
        \emph{Backends.} \textbf{I-GPU}: \texttt{icicle-snark} (GPU), \textbf{I-CPU}: \texttt{icicle-snark} (CPU), \textbf{S-CPU}: \texttt{snarkjs}, \textbf{R-CPU}: \texttt{rapidsnark}.\\
        \emph{Abbreviations.} \textbf{AF}: allocation failure.
    \end{minipage}

\end{table*}

\subsection{Scalability Evaluation}

Our protocol achieves constant per-round cost for training and commitment, making reward distribution the sole scalability bottleneck.
In this section, we evaluate how far \sys{} can scale in terms of blockchain deployment limits and off-chain proof cost.

\paragraph{Scalability Boundaries.}

We evaluate the system at scale by simulating reward distribution for up to 800 participants.
Ethereum enforces a maximum transaction gas limit of $2^{24}$ and restricts smart contract bytecode to 24KB.
While gas limits can be mitigated via batching, we find that the bytecode size is a stricter bottleneck.
Our deployment tests indicate that batch sizes exceeding 70 participants lead to contract failure regardless of available gas.
Therefore, the multi-shot variant must operate with batch sizes below this threshold to remain deployable.
As shown in Figure~\ref{fig:scalability-bytecode-size}, the bytecode size of the multi-shot variant depends on the per-batch size rather than the total number of participants.
This decoupling ensures the contract remains within the deployable bounds even under large workloads.
\begin{figure}[t]
    \centering
    \includegraphics[width=0.48\textwidth]{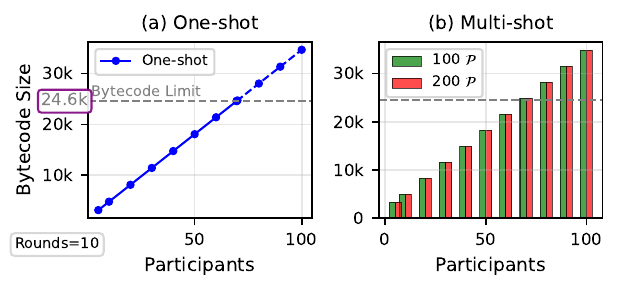}
    \caption{
        Bytecode size comparison between one-shot and multi-shot variants.
        \textnormal{The horizontal dashed line marks the 24KB maximum contract size allowed on Ethereum.}
    }
    \Description{}
    \label{fig:scalability-bytecode-size}
\end{figure}

\begin{figure}[t]
    \centering
    \includegraphics[width=0.48\textwidth]{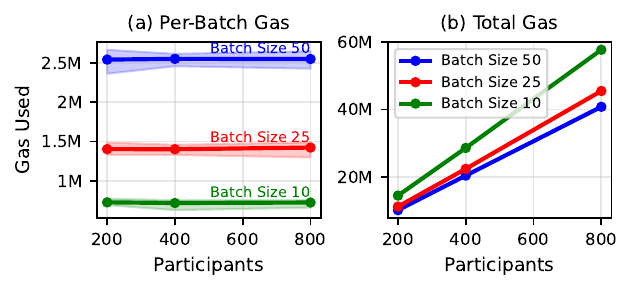}
    \caption{Multi-shot variant scalability up to 800 \pars.
        \textnormal{Per-batch gas is invariant to total \pars, while total gas scales linearly.}
    }
    \Description{Gas cost comparison between one-shot and multi-shot variants across varying batch sizes.}
    \label{fig:scalability-gas-cost}
\end{figure}

\paragraph{System Performance at Scale.}

Complementing the deployment analysis, Figure~\ref{fig:scalability-gas-cost} evaluates the corresponding on-chain gas consumption.
The results confirm the economic viability of the multi-shot variant at scale.
For instance, with a batch size of 50, the system requires 16 on-chain transactions to complete all 800 payouts, consuming less than 41M gas in total.
This configuration fits within Ethereum limits and represents a practical trade-off between contract complexity and the number of on-chain calls.

\paragraph{SNARK Proof Overhead.}

We evaluate the performance of SNARK proof generation across multiple backends, including both GPU- and CPU-based implementations.
For GPU-accelerated benchmarking, we use \texttt{icicle-snark}~\cite{soyturkIciclesnark2025}, tested on a server equipped with eight NVIDIA RTX 4090 GPUs.
Due to current implementation constraints, \texttt{icicle-snark} currently utilizes only a single GPU.
For CPU-based benchmarks, we include \texttt{rapidsnark}~\cite{lomakaRapidsnark2025} (C++), \texttt{snarkjs}~\cite{baylinaSnarkjs2018} (WebAssembly), and the CPU backend of \texttt{icicle-snark}.

Table~\ref{tab:zk-proof-benchmark} reports the average proof generation time over 10 runs for each backend and circuit.
For the configuration with 800 participants and 200 rounds, the largest circuit (used in the counter function) exceeds the available GPU memory under \texttt{icicle-snark}, resulting in an allocation failure (AF) error.
Nonetheless, for smaller configurations (e.g., 500 participants and 200 rounds), GPU-based backends complete proof generation within seconds, confirming the practical feasibility of our design under realistic settings.

\section{Conclusion and Future Work}
\label{sec:conclusion}

In this work, we presented \sys{}, a trustless and scalable reward settlement protocol that reconciles the high-frequency demands of FL with the scalability constraints of permissionless blockchains.
Guided by the \emph{Fundamental \sys{} Inequality}, \sys{} minimizes total economic friction through a dual-variant protocol design backed by a novel domain-specific circuit architecture.
By offering interoperable optimistic and validity-based settlement paths, the protocol flexibly adapts to varying cost-latency trade-offs.  Our deployment on the Ethereum Sepolia testnet confirms that \sys{} effectively breaks the scalability barrier, supporting 800 participants with negligible gas costs.
Future work includes enhancing reward fairness via advanced contribution estimation, comparing with other proving systems such as STARKs or recursive proofs, and leveraging historical incentive data to enhance participant reputation credibility in dynamic FL deployments.

\newpage

\bibliographystyle{ACM-Reference-Format}
\bibliography{ref}


\begin{thebibliography}{47}


\ifx \showCODEN    \undefined \def \showCODEN     #1{\unskip}     \fi
\ifx \showDOI      \undefined \def \showDOI       #1{#1}\fi
\ifx \showISBNx    \undefined \def \showISBNx     #1{\unskip}     \fi
\ifx \showISBNxiii \undefined \def \showISBNxiii  #1{\unskip}     \fi
\ifx \showISSN     \undefined \def \showISSN      #1{\unskip}     \fi
\ifx \showLCCN     \undefined \def \showLCCN      #1{\unskip}     \fi
\ifx \shownote     \undefined \def \shownote      #1{#1}          \fi
\ifx \showarticletitle \undefined \def \showarticletitle #1{#1}   \fi
\ifx \showURL      \undefined \def \showURL       {\relax}        \fi
\providecommand\bibfield[2]{#2}
\providecommand\bibinfo[2]{#2}
\providecommand\natexlab[1]{#1}
\providecommand\showeprint[2][]{arXiv:#2}

\bibitem[Alief et~al\mbox{.}(2023)]%
        {aliefFLB22023}
\bibfield{author}{\bibinfo{person}{Revin~Naufal Alief},
  \bibinfo{person}{Made~Adi Paramartha~Putra}, \bibinfo{person}{Augustin
  Gohil}, \bibinfo{person}{Jae-Min Lee}, {and} \bibinfo{person}{Dong-Seong
  Kim}.} \bibinfo{year}{2023}\natexlab{}.
\newblock \showarticletitle{{{FLB2}}: {{Layer}} 2 {{Blockchain Implementation
  Scheme}} on {{Federated Learning Technique}}}. In
  \bibinfo{booktitle}{\emph{2023 {{International Conference}} on {{Artificial
  Intelligence}} in {{Information}} and {{Communication}} ({{ICAIIC}})}}.
  \bibinfo{pages}{846--850}.
\newblock
\urldef\tempurl%
\url{https://doi.org/10.1109/ICAIIC57133.2023.10067038}
\showDOI{\tempurl}


\bibitem[{arnaucube}(2021)]%
        {arnaucubeKeccak256circom2021}
\bibfield{author}{\bibinfo{person}{{arnaucube}}.}
  \bibinfo{year}{2021}\natexlab{}.
\newblock \bibinfo{title}{Keccak256-Circom}.
\newblock \bibinfo{howpublished}{Vocdoni}.
\newblock
\urldef\tempurl%
\url{https://github.com/vocdoni/keccak256-circom}
\showURL{%
\tempurl}


\bibitem[Baylina and {iden3}(2018)]%
        {baylinaSnarkjs2018}
\bibfield{author}{\bibinfo{person}{Jordi Baylina} {and}
  \bibinfo{person}{{iden3}}.} \bibinfo{year}{2018}\natexlab{}.
\newblock \bibinfo{title}{Snarkjs: {{zkSNARK}} Implementation in {{JavaScript}}
  \& {{WASM}}}.
\newblock
\newblock
\urldef\tempurl%
\url{https://github.com/iden3/snarkjs}
\showURL{%
\tempurl}


\bibitem[{Bell{\'e}s-Mu{\~n}oz} et~al\mbox{.}(2023)]%
        {belles-munozCircom2023}
\bibfield{author}{\bibinfo{person}{Marta {Bell{\'e}s-Mu{\~n}oz}},
  \bibinfo{person}{Miguel Isabel}, \bibinfo{person}{Jose~Luis
  {Mu{\~n}oz-Tapia}}, \bibinfo{person}{Albert Rubio}, {and}
  \bibinfo{person}{Jordi Baylina}.} \bibinfo{year}{2023}\natexlab{}.
\newblock \showarticletitle{Circom: {{A Circuit Description Language}} for
  {{Building Zero-Knowledge Applications}}}.
\newblock \bibinfo{journal}{\emph{IEEE Trans. Dependable Secur. Comput.}}
  \bibinfo{volume}{20}, \bibinfo{number}{6} (\bibinfo{date}{Nov.}
  \bibinfo{year}{2023}), \bibinfo{pages}{4733--4751}.
\newblock
\urldef\tempurl%
\url{https://doi.org/10.1109/TDSC.2022.3232813}
\showDOI{\tempurl}


\bibitem[{Ben-Sasson} et~al\mbox{.}(2013)]%
        {ben-sassonSNARKs2013}
\bibfield{author}{\bibinfo{person}{Eli {Ben-Sasson}},
  \bibinfo{person}{Alessandro Chiesa}, \bibinfo{person}{Daniel Genkin},
  \bibinfo{person}{Eran Tromer}, {and} \bibinfo{person}{Madars Virza}.}
  \bibinfo{year}{2013}\natexlab{}.
\newblock \showarticletitle{{{SNARKs}} for {{C}}: {{Verifying Program
  Executions Succinctly}} and in {{Zero Knowledge}}}. In
  \bibinfo{booktitle}{\emph{Advances in {{Cryptology}} -- {{CRYPTO}} 2013}}.
  \bibinfo{publisher}{Springer}, \bibinfo{address}{Berlin, Heidelberg},
  \bibinfo{pages}{90--108}.
\newblock
\urldef\tempurl%
\url{https://doi.org/10.1007/978-3-642-40084-1_6}
\showDOI{\tempurl}


\bibitem[Bernstein et~al\mbox{.}(2011)]%
        {bernsteinHighSpeed2011}
\bibfield{author}{\bibinfo{person}{Daniel~J. Bernstein}, \bibinfo{person}{Niels
  Duif}, \bibinfo{person}{Tanja Lange}, \bibinfo{person}{Peter Schwabe}, {and}
  \bibinfo{person}{Bo-Yin Yang}.} \bibinfo{year}{2011}\natexlab{}.
\newblock \showarticletitle{High-{{Speed High-Security Signatures}}}. In
  \bibinfo{booktitle}{\emph{Cryptographic {{Hardware}} and {{Embedded Systems}}
  -- {{CHES}} 2011}}. \bibinfo{publisher}{Springer}, \bibinfo{address}{Berlin,
  Heidelberg}, \bibinfo{pages}{124--142}.
\newblock
\urldef\tempurl%
\url{https://doi.org/10.1007/978-3-642-23951-9_9}
\showDOI{\tempurl}


\bibitem[Caldas et~al\mbox{.}(2019)]%
        {caldasLEAF2019}
\bibfield{author}{\bibinfo{person}{Sebastian Caldas}, \bibinfo{person}{Sai
  Meher~Karthik Duddu}, \bibinfo{person}{Peter Wu}, \bibinfo{person}{Tian Li},
  \bibinfo{person}{Jakub Kone{\v c}n{\'y}}, \bibinfo{person}{H.~Brendan
  McMahan}, \bibinfo{person}{Virginia Smith}, {and} \bibinfo{person}{Ameet
  Talwalkar}.} \bibinfo{year}{2019}\natexlab{}.
\newblock \bibinfo{title}{{{LEAF}}: {{A Benchmark}} for {{Federated
  Settings}}}.
\newblock
\newblock
\urldef\tempurl%
\url{https://doi.org/10.48550/arXiv.1812.01097}
\showDOI{\tempurl}


\bibitem[Chen et~al\mbox{.}(2025)]%
        {chenFLock2025}
\bibfield{author}{\bibinfo{person}{Ruonan Chen}, \bibinfo{person}{Ye Dong},
  \bibinfo{person}{Yizhong Liu}, \bibinfo{person}{Tingyu Fan},
  \bibinfo{person}{Dawei Li}, \bibinfo{person}{Zhenyu Guan},
  \bibinfo{person}{Jianwei Liu}, {and} \bibinfo{person}{Jianying Zhou}.}
  \bibinfo{year}{2025}\natexlab{}.
\newblock \showarticletitle{{{FLock}}: {{Robust}} and {{Privacy-Preserving
  Federated Learning}} Based on {{Practical Blockchain State Channels}}}. In
  \bibinfo{booktitle}{\emph{Proceedings of the {{ACM}} on {{Web Conference}}
  2025}}. \bibinfo{publisher}{Association for Computing Machinery},
  \bibinfo{address}{New York, NY, USA}, \bibinfo{pages}{884--895}.
\newblock
\urldef\tempurl%
\url{https://doi.org/10.1145/3696410.3714666}
\showDOI{\tempurl}


\bibitem[Chen et~al\mbox{.}(2024)]%
        {chenContributions2024}
\bibfield{author}{\bibinfo{person}{Yiwei Chen}, \bibinfo{person}{Kaiyu Li},
  \bibinfo{person}{Guoliang Li}, {and} \bibinfo{person}{Yong Wang}.}
  \bibinfo{year}{2024}\natexlab{}.
\newblock \showarticletitle{Contributions {{Estimation}} in {{Federated
  Learning}}: {{A Comprehensive Experimental Evaluation}}}.
\newblock \bibinfo{journal}{\emph{Proceedings of the VLDB Endowment}}
  \bibinfo{volume}{17}, \bibinfo{number}{8} (\bibinfo{date}{April}
  \bibinfo{year}{2024}), \bibinfo{pages}{2077--2090}.
\newblock
\urldef\tempurl%
\url{https://doi.org/10.14778/3659437.3659459}
\showDOI{\tempurl}


\bibitem[{Costa-Gomes} et~al\mbox{.}(2001)]%
        {costa-gomesCognition2001}
\bibfield{author}{\bibinfo{person}{Miguel {Costa-Gomes}},
  \bibinfo{person}{Vincent~P. Crawford}, {and} \bibinfo{person}{Bruno
  Broseta}.} \bibinfo{year}{2001}\natexlab{}.
\newblock \showarticletitle{Cognition and {{Behavior}} in {{Normal-Form
  Games}}: {{An Experimental Study}}}.
\newblock \bibinfo{journal}{\emph{Econometrica}} \bibinfo{volume}{69},
  \bibinfo{number}{5} (\bibinfo{year}{2001}), \bibinfo{pages}{1193--1235}.
\newblock
\urldef\tempurl%
\url{https://doi.org/10.1111/1468-0262.00239}
\showDOI{\tempurl}


\bibitem[Deng(2012)]%
        {dengMNIST2012}
\bibfield{author}{\bibinfo{person}{Li Deng}.} \bibinfo{year}{2012}\natexlab{}.
\newblock \showarticletitle{The {{MNIST Database}} of {{Handwritten Digit
  Images}} for {{Machine Learning Research}} [{{Best}} of the {{Web}}]}.
\newblock \bibinfo{journal}{\emph{IEEE Signal Processing Magazine}}
  \bibinfo{volume}{29}, \bibinfo{number}{6} (\bibinfo{date}{Nov.}
  \bibinfo{year}{2012}), \bibinfo{pages}{141--142}.
\newblock
\urldef\tempurl%
\url{https://doi.org/10.1109/MSP.2012.2211477}
\showDOI{\tempurl}


\bibitem[Desai et~al\mbox{.}(2021)]%
        {desaiBlockFLA2021}
\bibfield{author}{\bibinfo{person}{Harsh~Bimal Desai},
  \bibinfo{person}{Mustafa~Safa Ozdayi}, {and} \bibinfo{person}{Murat
  Kantarcioglu}.} \bibinfo{year}{2021}\natexlab{}.
\newblock \showarticletitle{{{BlockFLA}}: {{Accountable Federated Learning}}
  via {{Hybrid Blockchain Architecture}}}. In
  \bibinfo{booktitle}{\emph{Proceedings of the {{Eleventh ACM Conference}} on
  {{Data}} and {{Application Security}} and {{Privacy}}}}.
  \bibinfo{publisher}{ACM}, \bibinfo{address}{Virtual Event USA},
  \bibinfo{pages}{101--112}.
\newblock
\urldef\tempurl%
\url{https://doi.org/10.1145/3422337.3447837}
\showDOI{\tempurl}


\bibitem[Dif et~al\mbox{.}(2025)]%
        {difAutoDFL2025}
\bibfield{author}{\bibinfo{person}{Meryem~Malak Dif},
  \bibinfo{person}{Mouhamed~Amine Bouchiha}, \bibinfo{person}{Mourad Rabah},
  {and} \bibinfo{person}{Yacine {Ghamri-Doudane}}.}
  \bibinfo{year}{2025}\natexlab{}.
\newblock \bibinfo{title}{{{AutoDFL}}: {{A Scalable}} and {{Automated
  Reputation-Aware Decentralized Federated Learning}}}.
\newblock
\newblock
\urldef\tempurl%
\url{https://doi.org/10.48550/arXiv.2501.04331}
\showDOI{\tempurl}


\bibitem[Dziembowski et~al\mbox{.}(2017)]%
        {dziembowskiPerun2017}
\bibfield{author}{\bibinfo{person}{Stefan Dziembowski}, \bibinfo{person}{Lisa
  Eckey}, \bibinfo{person}{Sebastian Faust}, {and} \bibinfo{person}{Daniel
  Malinowski}.} \bibinfo{year}{2017}\natexlab{}.
\newblock \bibinfo{title}{Perun: {{Virtual Payment Hubs}} over
  {{Cryptocurrencies}}}.
\newblock
\newblock
\urldef\tempurl%
\url{https://eprint.iacr.org/2017/635}
\showURL{%
\tempurl}


\bibitem[Fan et~al\mbox{.}(2021)]%
        {fanHybrid2021}
\bibfield{author}{\bibinfo{person}{Sizheng Fan}, \bibinfo{person}{Hongbo
  Zhang}, \bibinfo{person}{Yuchen Zeng}, {and} \bibinfo{person}{Wei Cai}.}
  \bibinfo{year}{2021}\natexlab{}.
\newblock \showarticletitle{Hybrid {{Blockchain-Based Resource Trading System}}
  for {{Federated Learning}} in {{Edge Computing}}}.
\newblock \bibinfo{journal}{\emph{IEEE Internet of Things Journal}}
  \bibinfo{volume}{8}, \bibinfo{number}{4} (\bibinfo{date}{Feb.}
  \bibinfo{year}{2021}), \bibinfo{pages}{2252--2264}.
\newblock
\urldef\tempurl%
\url{https://doi.org/10.1109/JIOT.2020.3028101}
\showDOI{\tempurl}


\bibitem[Gangwal et~al\mbox{.}(2023)]%
        {gangwalSurvey2023}
\bibfield{author}{\bibinfo{person}{Ankit Gangwal},
  \bibinfo{person}{Haripriya~Ravali Gangavalli}, {and} \bibinfo{person}{Apoorva
  Thirupathi}.} \bibinfo{year}{2023}\natexlab{}.
\newblock \showarticletitle{A Survey of {{Layer-two}} Blockchain Protocols}.
\newblock \bibinfo{journal}{\emph{Journal of Network and Computer
  Applications}}  \bibinfo{volume}{209} (\bibinfo{date}{Jan.}
  \bibinfo{year}{2023}), \bibinfo{pages}{103539}.
\newblock
\urldef\tempurl%
\url{https://doi.org/10.1016/j.jnca.2022.103539}
\showDOI{\tempurl}


\bibitem[Go et~al\mbox{.}(2009)]%
        {goTwitter2009}
\bibfield{author}{\bibinfo{person}{Alec Go}, \bibinfo{person}{Richa Bhayani},
  {and} \bibinfo{person}{Lei Huang}.} \bibinfo{year}{2009}\natexlab{}.
\newblock \bibinfo{booktitle}{\emph{Twitter {{Sentiment Classification}} Using
  {{Distant Supervision}}}}.
\newblock \bibinfo{type}{{T}echnical {R}eport}. \bibinfo{address}{CS224N
  project report, Stanford}.
\newblock


\bibitem[Goodfellow et~al\mbox{.}(2016)]%
        {goodfellowDeep2016}
\bibfield{author}{\bibinfo{person}{Ian Goodfellow}, \bibinfo{person}{Yoshua
  Bengio}, {and} \bibinfo{person}{Aaron Courville}.}
  \bibinfo{year}{2016}\natexlab{}.
\newblock \bibinfo{booktitle}{\emph{Deep Learning}}.
\newblock \bibinfo{publisher}{MIT Press}.
\newblock


\bibitem[Grassi et~al\mbox{.}(2021)]%
        {grassiPoseidon2021}
\bibfield{author}{\bibinfo{person}{Lorenzo Grassi}, \bibinfo{person}{Dmitry
  Khovratovich}, \bibinfo{person}{Christian Rechberger}, \bibinfo{person}{Arnab
  Roy}, {and} \bibinfo{person}{Markus Schofnegger}.}
  \bibinfo{year}{2021}\natexlab{}.
\newblock \showarticletitle{Poseidon: {{A New Hash Function}} for
  \textbraceleft{{Zero-Knowledge}}\textbraceright{} {{Proof Systems}}}. In
  \bibinfo{booktitle}{\emph{30th {{USENIX Security Symposium}} ({{USENIX
  Security}} 21)}}. \bibinfo{pages}{519--535}.
\newblock
\urldef\tempurl%
\url{https://www.usenix.org/conference/usenixsecurity21/presentation/grassi}
\showURL{%
\tempurl}


\bibitem[Groth(2016)]%
        {grothSize2016}
\bibfield{author}{\bibinfo{person}{Jens Groth}.}
  \bibinfo{year}{2016}\natexlab{}.
\newblock \bibinfo{title}{On the {{Size}} of {{Pairing-based Non-interactive
  Arguments}}}.
\newblock
\newblock
\urldef\tempurl%
\url{https://eprint.iacr.org/2016/260}
\showURL{%
\tempurl}


\bibitem[Hallaji et~al\mbox{.}(2024)]%
        {hallajiDecentralized2024}
\bibfield{author}{\bibinfo{person}{Ehsan Hallaji}, \bibinfo{person}{Roozbeh
  {Razavi-Far}}, \bibinfo{person}{Mehrdad Saif}, \bibinfo{person}{Boyu Wang},
  {and} \bibinfo{person}{Qiang Yang}.} \bibinfo{year}{2024}\natexlab{}.
\newblock \showarticletitle{Decentralized {{Federated Learning}}: {{A Survey}}
  on {{Security}} and {{Privacy}}}.
\newblock \bibinfo{journal}{\emph{IEEE Transactions on Big Data}}
  \bibinfo{volume}{10}, \bibinfo{number}{2} (\bibinfo{date}{April}
  \bibinfo{year}{2024}), \bibinfo{pages}{194--213}.
\newblock
\urldef\tempurl%
\url{https://doi.org/10.1109/TBDATA.2024.3362191}
\showDOI{\tempurl}


\bibitem[Han et~al\mbox{.}(2022)]%
        {hanHow2022}
\bibfield{author}{\bibinfo{person}{Rong Han}, \bibinfo{person}{Zheng Yan},
  \bibinfo{person}{Xueqin Liang}, {and} \bibinfo{person}{Laurence~T. Yang}.}
  \bibinfo{year}{2022}\natexlab{}.
\newblock \showarticletitle{How {{Can Incentive Mechanisms}} and {{Blockchain
  Benefit}} with {{Each Other}}? {{A Survey}}}.
\newblock \bibinfo{journal}{\emph{ACM Comput. Surv.}} \bibinfo{volume}{55},
  \bibinfo{number}{7} (\bibinfo{date}{Dec.} \bibinfo{year}{2022}),
  \bibinfo{pages}{136:1--136:38}.
\newblock
\urldef\tempurl%
\url{https://doi.org/10.1145/3539604}
\showDOI{\tempurl}


\bibitem[He et~al\mbox{.}(2016)]%
        {heDeep2016}
\bibfield{author}{\bibinfo{person}{Kaiming He}, \bibinfo{person}{Xiangyu
  Zhang}, \bibinfo{person}{Shaoqing Ren}, {and} \bibinfo{person}{Jian Sun}.}
  \bibinfo{year}{2016}\natexlab{}.
\newblock \showarticletitle{Deep {{Residual Learning}} for {{Image
  Recognition}}}. In \bibinfo{booktitle}{\emph{2016 {{IEEE Conference}} on
  {{Computer Vision}} and {{Pattern Recognition}} ({{CVPR}})}}.
  \bibinfo{pages}{770--778}.
\newblock
\urldef\tempurl%
\url{https://doi.org/10.1109/CVPR.2016.90}
\showDOI{\tempurl}


\bibitem[Hossain et~al\mbox{.}(2025)]%
        {hossainAdRoFL2025}
\bibfield{author}{\bibinfo{person}{Md~Kamrul Hossain}, \bibinfo{person}{Walid
  Aljoby}, \bibinfo{person}{Anis Elgabli}, \bibinfo{person}{Ahmed~M.
  Abdelmoniem}, {and} \bibinfo{person}{Khaled~A. Harras}.}
  \bibinfo{year}{2025}\natexlab{}.
\newblock \bibinfo{title}{{{AdRo-FL}}: {{Informed}} and {{Secure Client
  Selection}} for {{Federated Learning}} in the {{Presence}} of {{Adversarial
  Aggregator}}}.
\newblock
\newblock
\urldef\tempurl%
\url{https://doi.org/10.48550/arXiv.2506.17805}
\showDOI{\tempurl}


\bibitem[Huang et~al\mbox{.}(2024)]%
        {huangFederated2024}
\bibfield{author}{\bibinfo{person}{Wenke Huang}, \bibinfo{person}{Mang Ye},
  \bibinfo{person}{Zekun Shi}, \bibinfo{person}{Guancheng Wan},
  \bibinfo{person}{He Li}, \bibinfo{person}{Bo Du}, {and}
  \bibinfo{person}{Qiang Yang}.} \bibinfo{year}{2024}\natexlab{}.
\newblock \showarticletitle{Federated {{Learning}} for {{Generalization}},
  {{Robustness}}, {{Fairness}}: {{A Survey}} and {{Benchmark}}}.
\newblock \bibinfo{journal}{\emph{IEEE Transactions on Pattern Analysis and
  Machine Intelligence}} \bibinfo{volume}{46}, \bibinfo{number}{12}
  (\bibinfo{date}{Dec.} \bibinfo{year}{2024}), \bibinfo{pages}{9387--9406}.
\newblock
\urldef\tempurl%
\url{https://doi.org/10.1109/TPAMI.2024.3418862}
\showDOI{\tempurl}


\bibitem[Iacoban(2024)]%
        {iacobanStateFL2024}
\bibfield{author}{\bibinfo{person}{Ioana~Paula Iacoban}.}
  \bibinfo{year}{2024}\natexlab{}.
\newblock \emph{\bibinfo{title}{{{StateFL}}: {{State Channels Powered Federated
  Learning}}}}.
\newblock \bibinfo{thesistype}{Master's\ thesis}. \bibinfo{school}{Delft
  University of Technology}.
\newblock


\bibitem[Joulin et~al\mbox{.}(2017)]%
        {joulinBag2017}
\bibfield{author}{\bibinfo{person}{Armand Joulin}, \bibinfo{person}{Edouard
  Grave}, \bibinfo{person}{Piotr Bojanowski}, {and} \bibinfo{person}{Tomas
  Mikolov}.} \bibinfo{year}{2017}\natexlab{}.
\newblock \showarticletitle{Bag of {{Tricks}} for {{Efficient Text
  Classification}}}. In \bibinfo{booktitle}{\emph{Proceedings of the 15th
  {{Conference}} of the {{European Chapter}} of the {{Association}} for
  {{Computational Linguistics}}: {{Volume}} 2, {{Short Papers}}}}.
  \bibinfo{publisher}{Association for Computational Linguistics},
  \bibinfo{address}{Valencia, Spain}, \bibinfo{pages}{427--431}.
\newblock
\urldef\tempurl%
\url{https://aclanthology.org/E17-2068/}
\showURL{%
\tempurl}


\bibitem[Kipf and Welling(2017)]%
        {kipfSemiSupervised2017}
\bibfield{author}{\bibinfo{person}{Thomas~N. Kipf} {and} \bibinfo{person}{Max
  Welling}.} \bibinfo{year}{2017}\natexlab{}.
\newblock \showarticletitle{Semi-{{Supervised Classification}} with {{Graph
  Convolutional Networks}}}. In \bibinfo{booktitle}{\emph{International
  {{Conference}} on {{Learning Representations}}}}.
\newblock
\urldef\tempurl%
\url{https://openreview.net/forum?id=SJU4ayYgl}
\showURL{%
\tempurl}


\bibitem[Krizhevsky(2009)]%
        {krizhevskyLearning2009}
\bibfield{author}{\bibinfo{person}{Alex Krizhevsky}.}
  \bibinfo{year}{2009}\natexlab{}.
\newblock \showarticletitle{Learning {{Multiple Layers}} of {{Features}} from
  {{Tiny Images}}}.
\newblock  (\bibinfo{year}{2009}).
\newblock


\bibitem[Lecun et~al\mbox{.}(1998)]%
        {lecunGradientbased1998}
\bibfield{author}{\bibinfo{person}{Y. Lecun}, \bibinfo{person}{L. Bottou},
  \bibinfo{person}{Y. Bengio}, {and} \bibinfo{person}{P. Haffner}.}
  \bibinfo{year}{1998}\natexlab{}.
\newblock \showarticletitle{Gradient-Based Learning Applied to Document
  Recognition}.
\newblock \bibinfo{journal}{\emph{Proc. IEEE}} \bibinfo{volume}{86},
  \bibinfo{number}{11} (\bibinfo{date}{Nov.} \bibinfo{year}{1998}),
  \bibinfo{pages}{2278--2324}.
\newblock
\urldef\tempurl%
\url{https://doi.org/10.1109/5.726791}
\showDOI{\tempurl}


\bibitem[Liu et~al\mbox{.}(2024)]%
        {liuEnhancing2024}
\bibfield{author}{\bibinfo{person}{Ji Liu}, \bibinfo{person}{Chunlu Chen},
  \bibinfo{person}{Yu Li}, \bibinfo{person}{Lin Sun}, \bibinfo{person}{Yulun
  Song}, \bibinfo{person}{Jingbo Zhou}, \bibinfo{person}{Bo Jing}, {and}
  \bibinfo{person}{Dejing Dou}.} \bibinfo{year}{2024}\natexlab{}.
\newblock \showarticletitle{Enhancing Trust and Privacy in Distributed
  Networks: A Comprehensive Survey on Blockchain-Based Federated Learning}.
\newblock \bibinfo{journal}{\emph{Knowledge and Information Systems}}
  \bibinfo{volume}{66}, \bibinfo{number}{8} (\bibinfo{date}{Aug.}
  \bibinfo{year}{2024}), \bibinfo{pages}{4377--4403}.
\newblock
\urldef\tempurl%
\url{https://doi.org/10.1007/s10115-024-02117-3}
\showDOI{\tempurl}


\bibitem[Lomaka et~al\mbox{.}(2025)]%
        {lomakaRapidsnark2025}
\bibfield{author}{\bibinfo{person}{Oleh Lomaka}, \bibinfo{person}{Oleksandr
  Brezhniev}, {and} \bibinfo{person}{Jordi Baylina}.}
  \bibinfo{year}{2025}\natexlab{}.
\newblock \bibinfo{title}{Rapidsnark}.
\newblock \bibinfo{howpublished}{iden3}.
\newblock
\urldef\tempurl%
\url{https://github.com/iden3/rapidsnark}
\showURL{%
\tempurl}


\bibitem[Lyu et~al\mbox{.}(2020)]%
        {lyuCollaborative2020}
\bibfield{author}{\bibinfo{person}{Lingjuan Lyu}, \bibinfo{person}{Xinyi Xu},
  \bibinfo{person}{Qian Wang}, {and} \bibinfo{person}{Han Yu}.}
  \bibinfo{year}{2020}\natexlab{}.
\newblock \showarticletitle{Collaborative {{Fairness}} in {{Federated
  Learning}}}.
\newblock In \bibinfo{booktitle}{\emph{Federated {{Learning}}: {{Privacy}} and
  {{Incentive}}}}. \bibinfo{publisher}{Springer International Publishing},
  \bibinfo{address}{Cham}, \bibinfo{pages}{189--204}.
\newblock
\urldef\tempurl%
\url{https://doi.org/10.1007/978-3-030-63076-8_14}
\showDOI{\tempurl}


\bibitem[McCallum et~al\mbox{.}(2000)]%
        {mccallumAutomating2000}
\bibfield{author}{\bibinfo{person}{Andrew~Kachites McCallum},
  \bibinfo{person}{Kamal Nigam}, \bibinfo{person}{Jason Rennie}, {and}
  \bibinfo{person}{Kristie Seymore}.} \bibinfo{year}{2000}\natexlab{}.
\newblock \showarticletitle{Automating the {{Construction}} of {{Internet
  Portals}} with {{Machine Learning}}}.
\newblock \bibinfo{journal}{\emph{Information Retrieval}} \bibinfo{volume}{3},
  \bibinfo{number}{2} (\bibinfo{date}{July} \bibinfo{year}{2000}),
  \bibinfo{pages}{127--163}.
\newblock
\urldef\tempurl%
\url{https://doi.org/10.1023/A:1009953814988}
\showDOI{\tempurl}


\bibitem[Moro et~al\mbox{.}(2014)]%
        {moroDatadriven2014}
\bibfield{author}{\bibinfo{person}{S{\'e}rgio Moro}, \bibinfo{person}{Paulo
  Cortez}, {and} \bibinfo{person}{Paulo Rita}.}
  \bibinfo{year}{2014}\natexlab{}.
\newblock \showarticletitle{A Data-Driven Approach to Predict the Success of
  Bank Telemarketing}.
\newblock \bibinfo{journal}{\emph{Decision Support Systems}}
  \bibinfo{volume}{62} (\bibinfo{date}{June} \bibinfo{year}{2014}),
  \bibinfo{pages}{22--31}.
\newblock
\urldef\tempurl%
\url{https://doi.org/10.1016/j.dss.2014.03.001}
\showDOI{\tempurl}


\bibitem[Myerson(1992)]%
        {myersonValue1992}
\bibfield{author}{\bibinfo{person}{Roger~B. Myerson}.}
  \bibinfo{year}{1992}\natexlab{}.
\newblock \showarticletitle{On the {{Value}} of {{Game Theory}} in {{Social
  Science}}}.
\newblock \bibinfo{journal}{\emph{Rationality and Society}}
  \bibinfo{volume}{4}, \bibinfo{number}{1} (\bibinfo{year}{1992}),
  \bibinfo{pages}{62--73}.
\newblock
\urldef\tempurl%
\url{https://ideas.repec.org//a/sae/ratsoc/v4y1992i1p62-73.html}
\showURL{%
\tempurl}


\bibitem[Nakamoto(2008)]%
        {nakamotoBitcoin2008}
\bibfield{author}{\bibinfo{person}{Satoshi Nakamoto}.}
  \bibinfo{year}{2008}\natexlab{}.
\newblock \showarticletitle{Bitcoin: {{A Peer-to-Peer Electronic Cash
  System}}}.
\newblock  (\bibinfo{year}{2008}), \bibinfo{pages}{9}.
\newblock
\urldef\tempurl%
\url{https://bitcoin.org/bitcoin.pdf}
\showURL{%
\tempurl}


\bibitem[Reno and Roy(2025)]%
        {renoNavigating2025}
\bibfield{author}{\bibinfo{person}{Saha Reno} {and} \bibinfo{person}{Koushik
  Roy}.} \bibinfo{year}{2025}\natexlab{}.
\newblock \showarticletitle{Navigating the {{Blockchain Trilemma}}: {{A
  Review}} of {{Recent Advances}} and {{Emerging Solutions}} in
  {{Decentralization}}, {{Security}}, and {{Scalability Optimization}}}.
\newblock \bibinfo{journal}{\emph{Computers, Materials \& Continua}}
  \bibinfo{volume}{84}, \bibinfo{number}{2} (\bibinfo{year}{2025}),
  \bibinfo{pages}{2061--2119}.
\newblock
\urldef\tempurl%
\url{https://doi.org/10.32604/cmc.2025.066366}
\showDOI{\tempurl}


\bibitem[Sguanci et~al\mbox{.}(2021)]%
        {sguanciLayer2021}
\bibfield{author}{\bibinfo{person}{Cosimo Sguanci}, \bibinfo{person}{Roberto
  Spatafora}, {and} \bibinfo{person}{Andrea~Mario Vergani}.}
  \bibinfo{year}{2021}\natexlab{}.
\newblock \bibinfo{title}{Layer 2 {{Blockchain Scaling}}: A {{Survey}}}.
\newblock
\newblock
\urldef\tempurl%
\url{https://doi.org/10.48550/arXiv.2107.10881}
\showDOI{\tempurl}


\bibitem[Smilkov et~al\mbox{.}(2019)]%
        {smilkovTensorFlowjs2019}
\bibfield{author}{\bibinfo{person}{Daniel Smilkov}, \bibinfo{person}{Nikhil
  Thorat}, \bibinfo{person}{Yannick Assogba}, \bibinfo{person}{Ann Yuan},
  \bibinfo{person}{Nick Kreeger}, \bibinfo{person}{Ping Yu},
  \bibinfo{person}{Kangyi Zhang}, \bibinfo{person}{Shanqing Cai},
  \bibinfo{person}{Eric Nielsen}, \bibinfo{person}{David Soergel},
  \bibinfo{person}{Stan Bileschi}, \bibinfo{person}{Michael Terry},
  \bibinfo{person}{Charles Nicholson}, \bibinfo{person}{Sandeep~N. Gupta},
  \bibinfo{person}{Sarah Sirajuddin}, \bibinfo{person}{D. Sculley},
  \bibinfo{person}{Rajat Monga}, \bibinfo{person}{Greg Corrado},
  \bibinfo{person}{Fernanda~B. Vi{\'e}gas}, {and} \bibinfo{person}{Martin
  Wattenberg}.} \bibinfo{year}{2019}\natexlab{}.
\newblock \bibinfo{title}{{{TensorFlow}}.Js: {{Machine Learning}} for the
  {{Web}} and {{Beyond}}}.
\newblock
\newblock
\urldef\tempurl%
\url{https://doi.org/10.48550/arXiv.1901.05350}
\showDOI{\tempurl}


\bibitem[Soyturk(2025)]%
        {soyturkIciclesnark2025}
\bibfield{author}{\bibinfo{person}{Emir Soyturk}.}
  \bibinfo{year}{2025}\natexlab{}.
\newblock \bibinfo{title}{Icicle-Snark}.
\newblock \bibinfo{howpublished}{Ingonyama}.
\newblock
\urldef\tempurl%
\url{https://github.com/ingonyama-zk/icicle-snark}
\showURL{%
\tempurl}


\bibitem[Wahrst{\"a}tter et~al\mbox{.}(2024)]%
        {wahrstatterOpenFL2024}
\bibfield{author}{\bibinfo{person}{Anton Wahrst{\"a}tter},
  \bibinfo{person}{Sajjad Khan}, {and} \bibinfo{person}{Davor Svetinovic}.}
  \bibinfo{year}{2024}\natexlab{}.
\newblock \showarticletitle{{{OpenFL}}: {{A}} Scalable and Secure Decentralized
  Federated Learning System on the {{Ethereum}} Blockchain}.
\newblock \bibinfo{journal}{\emph{Internet of Things}}  \bibinfo{volume}{26}
  (\bibinfo{date}{July} \bibinfo{year}{2024}), \bibinfo{pages}{101174}.
\newblock
\urldef\tempurl%
\url{https://doi.org/10.1016/j.iot.2024.101174}
\showDOI{\tempurl}


\bibitem[Wang et~al\mbox{.}(2020)]%
        {wangPrincipled2020}
\bibfield{author}{\bibinfo{person}{Tianhao Wang}, \bibinfo{person}{Johannes
  Rausch}, \bibinfo{person}{Ce Zhang}, \bibinfo{person}{Ruoxi Jia}, {and}
  \bibinfo{person}{Dawn Song}.} \bibinfo{year}{2020}\natexlab{}.
\newblock \showarticletitle{A {{Principled Approach}} to {{Data Valuation}} for
  {{Federated Learning}}}.
\newblock In \bibinfo{booktitle}{\emph{Federated {{Learning}}: {{Privacy}} and
  {{Incentive}}}}. \bibinfo{publisher}{Springer International Publishing},
  \bibinfo{address}{Cham}, \bibinfo{pages}{153--167}.
\newblock
\urldef\tempurl%
\url{https://doi.org/10.1007/978-3-030-63076-8_11}
\showDOI{\tempurl}


\bibitem[Weng et~al\mbox{.}(2021)]%
        {wengDeepChain2021}
\bibfield{author}{\bibinfo{person}{Jiasi Weng}, \bibinfo{person}{Jian Weng},
  \bibinfo{person}{Jilian Zhang}, \bibinfo{person}{Ming Li},
  \bibinfo{person}{Yue Zhang}, {and} \bibinfo{person}{Weiqi Luo}.}
  \bibinfo{year}{2021}\natexlab{}.
\newblock \showarticletitle{{{DeepChain}}: {{Auditable}} and
  {{Privacy-Preserving Deep Learning}} with {{Blockchain-Based Incentive}}}.
\newblock \bibinfo{journal}{\emph{IEEE Transactions on Dependable and Secure
  Computing}} \bibinfo{volume}{18}, \bibinfo{number}{5} (\bibinfo{date}{Sept.}
  \bibinfo{year}{2021}), \bibinfo{pages}{2438--2455}.
\newblock
\urldef\tempurl%
\url{https://doi.org/10.1109/TDSC.2019.2952332}
\showDOI{\tempurl}


\bibitem[Xiao et~al\mbox{.}(2017)]%
        {xiaoFashionMNIST2017}
\bibfield{author}{\bibinfo{person}{Han Xiao}, \bibinfo{person}{Kashif Rasul},
  {and} \bibinfo{person}{Roland Vollgraf}.} \bibinfo{year}{2017}\natexlab{}.
\newblock \bibinfo{title}{Fashion-{{MNIST}}: A {{Novel Image Dataset}} for
  {{Benchmarking Machine Learning Algorithms}}}.
\newblock
\newblock
\urldef\tempurl%
\url{https://doi.org/10.48550/arXiv.1708.07747}
\showDOI{\tempurl}


\bibitem[Xing et~al\mbox{.}(2026)]%
        {xingZeroKnowledge2026}
\bibfield{author}{\bibinfo{person}{Zhibo Xing}, \bibinfo{person}{Zijian Zhang},
  \bibinfo{person}{Ziang Zhang}, \bibinfo{person}{Zhen Li},
  \bibinfo{person}{Meng Li}, \bibinfo{person}{Jiamou Liu},
  \bibinfo{person}{Zongyang Zhang}, \bibinfo{person}{Yi Zhao},
  \bibinfo{person}{Qi Sun}, \bibinfo{person}{Liehuang Zhu}, {and}
  \bibinfo{person}{Giovanni Russello}.} \bibinfo{year}{2026}\natexlab{}.
\newblock \showarticletitle{Zero-{{Knowledge Proof-based Verifiable
  Decentralized Machine Learning}} in {{Communication Network}}: {{A
  Comprehensive Survey}}}.
\newblock \bibinfo{journal}{\emph{IEEE Communications Surveys \& Tutorials}}
  \bibinfo{volume}{28} (\bibinfo{year}{2026}), \bibinfo{pages}{985--1024}.
\newblock
\urldef\tempurl%
\url{https://doi.org/10.1109/COMST.2025.3561657}
\showDOI{\tempurl}


\bibitem[Xu et~al\mbox{.}(2021)]%
        {xuBAFL2021}
\bibfield{author}{\bibinfo{person}{Chenhao Xu}, \bibinfo{person}{Youyang Qu},
  \bibinfo{person}{Peter~W. Eklund}, \bibinfo{person}{Yong Xiang}, {and}
  \bibinfo{person}{Longxiang Gao}.} \bibinfo{year}{2021}\natexlab{}.
\newblock \showarticletitle{{{BAFL}}: {{An Efficient Blockchain-Based
  Asynchronous Federated Learning Framework}}}. In
  \bibinfo{booktitle}{\emph{2021 {{IEEE Symposium}} on {{Computers}} and
  {{Communications}} ({{ISCC}})}}. \bibinfo{pages}{1--6}.
\newblock
\urldef\tempurl%
\url{https://doi.org/10.1109/ISCC53001.2021.9631405}
\showDOI{\tempurl}


\end{thebibliography}

\appendix
\section{Circuit Gadgets}
\label{appendix:verifier_gadgets}
\begin{algorithm}[ht]
    \caption{Gadgets}\label{alg:gadgets}
    \small

    \SetKwProg{func}{gadget}{}{}

    \func{$\gadgetname{LPI}$($x$) $\to (\vec{y})$ \CommentFunction{$\vec{y} = (\{1\}^x, \{0\}^{T-x})$}} {
        $\forall i \in [T],\quad y_i \gets \mathbf{1}[i<x]$\;

        \constrain{\sum_{i=0}^{T-1} y_i = x}
        \Comment*[r]{sum equals the integer}

        $\forall i \in [T]$,
        \constrain{y_i\cdot(y_i-1) = 0}
        \Comment*[r]{booleanity}

        $\forall i \in [T+1]$,
        \constrain{c_i \gets
            \begin{cases}
                1 - y_i & \text{if }i=0\\
                y_{i-1} & \text{if }i=T\\
                y_{i-1} - y_i & \text{otherwise}
            \end{cases}
        }\;

        \constrain{\sum_{i=0}^{T} c_i = 1}
        \Comment*[r]{exactly one $1\!\to\!0$ transition}
    }

    \func{$\gadgetname{IsZero}$($x$) $\to (y)$ \CommentFunction{$y = \mathbf{1}[x=0]$}}  {
        $\mathit{inv} \gets (x = 0\text{ ? }0 : x^{-1})$\;

        \constrain{x\cdot \mathit{inv} = 1 - y},
        \constrain{x\cdot y = 0}\;
    }

    \func{$H_{M}(\vec{x}, \salt)$ $\to (y)$ \CommentFunction{Modular hash}} {
        $k \gets \left\lceil \frac{N}{\kappa-1} \right\rceil$, $y \gets \salt$\;
        \For{$i \gets 0$ \textbf{to} $k-1$}{
            $m \gets \min(N- i\cdot (\kappa-1) , \kappa-1)$\;

            $\forall j \in [m]$,
            \constrain{s_j \gets x_{i\cdot (\kappa-1) + j}}\;

            \constrain{y \gets H_{m+1} (y, s_0,\dots,s_{m-1})}\;
        }
    }

    \func{$\gadgetname{BatchExtractor}$($\vec{x}, a$) $\to (\vec{y})$ \CommentFunction{$\vec{y} = \{x_{a B}, \ldots, x_{a B + B - 1}\}$}} {

        $t \gets \lceil N/B \rceil$\;

        $\forall i \in [t],\quad$ \constrain{e_i \gets \textsc{IsZero}(a-i)}\;

        \For{$k \gets 0$ \textbf{to} $t \cdot B-1$} {
            $i \gets \lfloor k/B \rfloor$, \quad $j \gets k \bmod B$\;

            \constrain{S[i][j] \gets
                \begin{cases}
                    x_k\cdot e_i & (k<N)\\
                    0 & (\text{otherwise})
                \end{cases}
            }\;
        }

        $\forall j \in [B]$,
        \constrain{y_j \gets \sum_{i=0}^{t-1} S[i][j]}\;
    }

    \func{$\gadgetname{SumRow}$($\bm{X},\vec{r}$) $\to (\vec{y})$} {
        $\forall i \in [n],\quad$
        \constrain{y_i \gets \sum_{t=0}^{T-1} \bm{X}[i][t] \cdot r_t}\;
    }

    \func{$\gadgetname{MaskCheck}$($\bm{X},\vec{r}$)} {
        $\forall (i,t) \in [n]\times[T],\quad$
        \constrain{\bm{X}[i][t] \cdot (1-r_t) = 0}\;
    }

    \func{$H_\text{PK}$($\mathbf{A}$) $\to (\hagg)$} {
        \constrain{\hagg \gets H_2(\mathbf{A}_x,\mathbf{A}_y)}\;
    }

    \func{$\gadgetname{VerifyCommitment}$($\bm{V},\pars,\salt,\mathbf{A},\sigma_C$) $\to (C)$} {
        \constrain{C \gets H_3(H_{M}(\text{vec}(\bm{V})),\;
        H_{M}(\pars),\;\salt)}\;

        \constrain{\text{EdDSA.verify}(\mathbf{A},\sigma_C,C) = 1}\;
    }

\end{algorithm}

\newpage

\section{Main Verifiers}
\label{appendix:verifier_circuits}
\begin{algorithm}[ht]
    \caption{Main Verifier Circuits}\label{alg:verifiers}
    \small

    \SetKwProg{func}{circuit}{}{}

    \func{$\verify{transition}$($\pubpar{r},\bm{V}_2,\vec{p}_2,p_1,\salt,\mathbf{A},\sigma_{C_1},\sigma_{C_2}$) $\to (\pubpar{C_1},\pubpar{C_2},\pubpar{\hagg})$} {
        \constrain{\vec{r}_2 \gets \gadgetname{LPI}(\pubpar{r}+1)}\;
        \constrain{\gadgetname{MaskCheck}(\vec{r}_2,\bm{V}_2)}
        \Comment*[r]{Verify $\bm{V}_2$ belongs to round $\pubpar{r}+1$}
        \BlankLine

        \constrain{\vec{r}_1 \gets \gadgetname{LPI}(\pubpar{r})}
        \CommentDown*[r]{Derive $\bm{V}_1$ from $\bm{V}_2$ (round $\pubpar{r}$ slice)}

        $\forall (i,t) \in [n]\times[T]$,
        \constrain{\bm{V}_1[i][t] \gets \bm{V}_2[i][t] \cdot \vec{r}_1[t]}\;
        \BlankLine

        \constrain{\vec{r}_{p_1} \gets \gadgetname{LPI}(p_1)}\;
        $\forall i \in [n]$, \constrain{\vec{p}_1[i] \gets \vec{p}_2[i] \cdot \vec{r}_{p_1}[i]}
        \Comment*[r]{Derive $\vec{p}_1$ from $\vec{p}_2$}
        \BlankLine

        \constrain{\pubpar{\hagg} \gets H_\text{PK}(\mathbf{A})}\;

        $\forall i \in \{1,2\}$, \constrain{\pubpar{C_i} \gets \gadgetname{VerifyCommitment}(\bm{V}_i, \vec{p}_i, \salt, \mathbf{A}, \sigma_{C_i})}\;
    }

    \func{$\verify{challenge}$($\pubpar{r},\bm{V},\vec{p},\salt,\mathbf{A},\sigma_{C}$) $\to (\pubpar{C},\pubpar{\hagg})$} {
        \constrain{\vec{r} \gets \gadgetname{LPI}(\pubpar{r})}\;
        \constrain{\gadgetname{MaskCheck}(\vec{r},\bm{V})}
        \Comment*[r]{Verify $\bm{V}$ belongs to round $\pubpar{r}+1$}
        \BlankLine

        \constrain{\pubpar{\hagg} \gets H_\text{PK}(\mathbf{A})}\;

        \constrain{\pubpar{C} \gets \gadgetname{VerifyCommitment}(\bm{V}, \vec{p}, \salt, \mathbf{A}, \sigma_{C})}\;
    }

    \func{$\verify{distribute}$($\pubpar{r},\pubpar{b},\bm{V},\vec{p},\salt,\mathbf{A},\sigma_{C}$) $\to (\pubpar{C},\pubpar{\hagg}, \pubpar{\vec{s}_p}, \pubpar{\vec{p}_p})$} {
        \constrain{\vec{r} \gets \gadgetname{LPI}(\pubpar{r})}\;
        \constrain{\gadgetname{MaskCheck}(\vec{r},\bm{V})}
        \Comment*[r]{Verify $\bm{V}$ belongs to round $\pubpar{r}+1$}
        \BlankLine

        \constrain{\pubpar{\vec{s}} \gets \gadgetname{SumRow}(\bm{V},\vec{r})}\;
        \constrain{\pubpar{\vec{p}_p} \gets \gadgetname{BatchExtractor}(\vec{p},\pubpar{b})}\;
        \constrain{\pubpar{\vec{s}_p} \gets \gadgetname{BatchExtractor}(\vec{s},\pubpar{b})}\;
        \BlankLine

        \constrain{\pubpar{\hagg} \gets H_\text{PK}(\mathbf{A})}\;

        \constrain{\pubpar{C} \gets \gadgetname{VerifyCommitment}(\bm{V}, \vec{p}, \salt, \mathbf{A}, \sigma_{C})}\;
    }

\end{algorithm}

\section{Incentive Compatibility}
\label{appendix:incentive}

We formalize the incentive compatibility of our protocol by modeling the interaction as a static complete-information game.
Each player selects a strategy to maximize its utility, and we show that honest behavior constitutes a Nash equilibrium~\cite{myersonValue1992} under rational assumptions.

\paragraph{Game Model.}

We define a one-shot normal-form game~\cite{costa-gomesCognition2001}:
$$
\mathcal{G} = (M, \{\mathcal{S}_i\}_{i \in M}, \{u_i\}_{i \in M}),
$$
where $M = \{\agg, P_1, \dots, P_N\}$ is the set of all players, consisting of the aggregator $\agg$ and participants $P_1, \dots, P_N$.
Each player $i \in M$ chooses a strategy $s_i \in \mathcal{S}_i$ to maximize its utility $u_i(s_i, s_{-i})$.

Although the overall protocol may consist of multiple training rounds, we model the interaction as a one-shot static game because the first occurrence of malicious behavior (e.g., incorrect commitment or false challenge) triggers an immediate resolution phase involving dispute handling, slashing, and final reward distribution.
This effectively terminates the ongoing round-based process, making the game finite and strategically equivalent to a single-shot interaction.

Furthermore, as participants interact only with the aggregator and not with each other, their incentives can be analyzed independently.
The strategy of each participant depends only on the aggregator's behavior and does not influence or depend on the actions of other participants.

\emph{Aggregator strategy space $\mathcal{S}_{\agg}$:}
\begin{itemize}
    \item $s_{\agg}^{\text{honest}}$: Commit correct rewards and finalize distribution.
    \item $s_{\agg}^{\text{tamper}}$: Publish incorrect rewards for extra gain.
    \item $s_{\agg}^{\text{abort}}$: Exit early to avoid costs.
\end{itemize}

\emph{Participant strategy space $\mathcal{S}_{P_i}$:}
\begin{itemize}
    \item $s_i^{\text{honest}}$: Train, verify, and act accordingly.
    \item $s_i^{\text{malicious}}$: Submit false challenge to steal rewards.
    \item $s_i^{\text{passive}}$: Skip verification to avoid cost.
\end{itemize}

\paragraph{Utility Function.}
To facilitate formal analysis, we explicitly define all associated parameters for the utility functions introduced.
Each player's utility is determined by reward $R$, penalty $P$, and cost $C$, with binary indicators $I$ determining which terms apply.

\emph{Aggregator \agg:}
$$
u_{\agg} =
\begin{cases}
    R_{\agg}^{\text{model}} + R_{\agg}^{\text{bonus}} - C_{\agg}^{\text{commit}} & \text{if } s_{\agg} = s_{\agg}^{\text{honest}} \\
    R_{\agg}^{\text{model}} - C_{\agg}^{\text{commit}} - I_{\agg}^{\text{caught}} \cdot P_{\agg}^{\text{slash}} & \text{if } s_{\agg} = s_{\agg}^{\text{tamper}} \\
    - I_{\agg}^{\text{caught}} \cdot P_{\agg}^{\text{slash}} & \text{if } s_{\agg} = s_{\agg}^{\text{abort}}
\end{cases}
\,,
$$

\emph{Participant $P_i$:}
$$
u_{P_i} =
\begin{cases}
    R_i^{\text{reward}} + I_i^{\text{chal}} \cdot P_{\agg}^{\text{slash}} - C_i^{\text{gas}}  & \text{if } s_{P_i} = s_{P_i}^{\text{honest}} \\
    \begin{aligned}[t]
        & R_i^{\text{reward}} + (1 - I_i^{\text{fail}}) \cdot R_i^{\text{steal}} \\
        & \quad - I_i^{\text{fail}} \cdot P_i^{\text{slash}} - C_i^{\text{gas}}
    \end{aligned} & \text{if } s_{P_i} = s_{P_i}^{\text{malicious}} \\
    R_i^{\text{reward}} & \text{if } s_{P_i} = s_{P_i}^{\text{passive}}
\end{cases}
\,,
$$

\noindent \emph{where:}
\begin{itemize}
    \item $R_{\agg}^{\text{bonus}}$: reward for finalizing reward distribution.
    \item $R_{\agg}^{\text{model}}$: reward for \agg achieving the desired final model.
    \item $R_i^{\text{reward}}$: reward received by $P_i$ as incentive.
    \item $R_i^{\text{steal}}$: reward gained by $P_i$ through successful theft.
    \item $P_{\agg}^{\text{slash}}, P_i^{\text{slash}}$: penalty for caught misbehavior.
    \item $C_{\agg}^{\text{commit}}, C_{i}^{\text{gas}}$: cost for on-chain operations.
    \item $I_{\agg}^{\text{caught}}$: indicator for whether \agg tampering is challenged by at least one participant.
    \item $I_i^{\text{chal}}$: indicator for whether $P_i$ challenges upon detecting misbehavior.
    \item $I_i^{\text{fail}}$: indicator for whether $P_i$'s challenge fails.
\end{itemize}

\begin{proof}[Proof of Theorem 1]

    Here we provide the formal proof for Theorem 1 presented in Section \ref{sec:incentive}.
    We prove that the strategy profile
    $$
    \mathbf{s}^* = (s_{\agg}^{\text{honest}}, s_1^{\text{honest}}, \dots, s_N^{\text{honest}})
    $$
    constitutes a Nash equilibrium of the game $\mathcal{G} = (M, \{\mathcal{S}_i\}_{i \in M}, \{u_i\}_{i \in M})$ under rational behavior.

    \paragraph{Binary Indicator Instantiations.}
    In the strategy profile $\mathbf{s}^*$, we have:
    \begin{itemize}
        \item $I_{\agg}^{\text{caught}} = 1$, since all $P_i$ are honest (i.e., not passive).
        \item $I_i^{\text{chal}} = 1$, as $P_i$ verifies and challenges when appropriate.
        \item $I_i^{\text{fail}} = 1$, since \agg is honest and counters all false challenges.
    \end{itemize}

    \paragraph{Aggregator Deviation Analysis.}
    Assume all \pars follow $s_i^{\text{honest}}$.
    Then:

    \begin{itemize}
        \item Honest behavior yields:
            $$
            u_{\agg}^{\text{honest}} = R_{\agg}^{\text{model}} + R_{\agg}^{\text{bonus}} - C_{\agg}^{\text{commit}}.
            $$
        \item Tampering yields:
            $$
            u_{\agg}^{\text{tamper}} = R_{\agg}^{\text{model}} - C_{\agg}^{\text{commit}} - P_{\agg}^{\text{slash}}.
            $$
        \item Aborting yields:
            $$
            u_{\agg}^{\text{abort}} = - P_{\agg}^{\text{slash}}.
            $$
    \end{itemize}

    Since all parameters are positive, we have:
    $$
    u_{\agg}^{\text{honest}} > u_{\agg}^{\text{tamper}} \quad \text{and} \quad u_{\agg}^{\text{honest}} > u_{\agg}^{\text{abort}}.
    $$
    Thus, \agg has no incentive to deviate.

    \paragraph{Participant Deviation Analysis.}
    Fix $s_{\agg} = s_{\agg}^{\text{honest}}$ and $s_j = s_j^{\text{honest}}$ for all $j \ne i$.
    Consider $P_i$:

    \begin{itemize}
        \item Honest behavior yields:
            $$
            u_i^{\text{honest}} = R_i^{\text{reward}} + P_{\agg}^{\text{slash}} - C_i^{\text{gas}}.
            $$
        \item Malicious behavior yields:
            $$
            u_i^{\text{malicious}} = R_i^{\text{reward}} - P_i^{\text{slash}} - C_i^{\text{gas}}.
            $$
            Since all parameters are positive, we have:
            $$
            u_i^{\text{honest}} > u_i^{\text{malicious}}.
            $$
        \item Passive behavior yields:
            $$
            u_i^{\text{passive}} = R_i^{\text{reward}}.
            $$
            Honest behavior is strictly better if:
            $$
            u_i^{\text{honest}} > u_i^{\text{passive}} \implies P_{\agg}^{\text{slash}} > C_i^{\text{gas}}.
            $$
    \end{itemize}

    \paragraph{Conclusion.}
    No player has an incentive to deviate from $\mathbf{s}^*$ if the following condition holds:
    $$
    P_{\agg}^{\text{slash}} > C_i^{\text{gas}} \quad \text{for all } i \in M.
    $$
    Therefore, $\mathbf{s}^*$ is a Nash equilibrium of $\mathcal{G}$.
\end{proof}

\paragraph{Example.}
We illustrate the above analysis with a concrete parameter setting:

\begin{itemize}
    \item Aggregator: $R_{\agg}^{\text{model}} = 10$, $R_{\agg}^{\text{bonus}} = 4$, $C_{\agg}^{\text{commit}} = 1$, $P_{\agg}^{\text{slash}} = 3$.
    \item Participant $P_i$: $R_i^{\text{reward}} = 2$, $R_i^{\text{steal}} = 2$, $C_i^{\text{gas}} = 1$, $P_i^{\text{slash}} = 2$.
\end{itemize}

Under the strategy profile $\mathbf{s}^*$ where all players behave honestly, we compute:

\begin{itemize}
    \item Aggregator utility:
        $$
        \begin{aligned}[t]
            u_{\agg}^{\text{honest}} & = 10 + 4 - 1 = 13,
            \\
            u_{\agg}^{\text{tamper}} & = 10 - 1 - 3 = 6,
            \\
            u_{\agg}^{\text{abort}} & = -3.
        \end{aligned}
        $$
        The aggregator strictly prefers honest behavior.
    \item Participant utility:
        $$
        \begin{aligned}[t]
            u_i^{\text{honest}} & = 2 + 3 - 1 = 4,
            \\
            u_i^{\text{malicious}} & = 2 - 2 - 1 = -1,
            \\
            u_i^{\text{passive}} & = 2.
        \end{aligned}
        $$
        Honest behavior yields the highest utility for each $P_i$.
    \item The equilibrium condition $P_{\agg}^{\text{slash}} > C_i^{\text{gas}}$ is satisfied, as $3 > 1$.
\end{itemize}

Hence, no player has an incentive to deviate, and the strategy profile $\mathbf{s}^*$ constitutes a Nash equilibrium.

\end{document}